\documentclass[onecolumn,11pt,oneside,draftclsnofoot]{IEEEtran}%

\usepackage{amsbsy}
\usepackage{floatflt} 

\usepackage{amsmath}
\usepackage{amssymb}
\usepackage{times}
\usepackage{graphicx}
\usepackage{xspace}
\usepackage{paralist} 
\usepackage{setspace} 
\usepackage{xypic}
\xyoption{curve}
\usepackage{latexsym}
\usepackage{theorem}
\usepackage{ifthen}
\usepackage{subfigure}
\usepackage{turnstile}

\usepackage{booktabs}
{\theoremheaderfont{\it} \theorembodyfont{\rmfamily}
\newtheorem{theorem}{Theorem}
\newtheorem{lemma}[theorem]{Lemma}
\newtheorem{definition}[theorem]{Definition}

\newtheorem{proposition}[theorem]{Proposition}

\newtheorem{result}[theorem]{Result}

}

\title{Kalman Filtering with Intermittent Observations: Weak Convergence to a Stationary Distribution}
\author{Soummya Kar, Bruno Sinopoli, and Jos\'e M.~F.~Moura$^{*}$
\thanks{The authors are with the Dep.~Electrical and Computer Engineering, Carnegie Mellon University, Pittsburgh, PA 15213, USA (e-mail:
soummyak@andrew.cmu.edu, brunos@ece.cmu.edu, moura@ece.cmu.edu, ph: (412)268-6341, fax: (412)268-3890.)}
\thanks{Work supported by NSF under grants \#~ECS-0225449
and~\#~CNS-0428404.}}
\begin{document}
\maketitle\thispagestyle{empty}\maketitle

\begin{abstract}
The paper studies the asymptotic behavior of discrete time Random Riccati Equations~(RRE) arising in Kalman filtering when the
arrival of the observations is described by a Bernoulli
i.i.d.~process. We model the RRE as an order-preserving, strongly
sublinear random dynamical system~(RDS). Under a sufficient
condition, stochastic boundedness, and using a limit-set dichotomy
result for order-preserving, strongly sublinear RDS, we establish
the asymptotic properties of the RRE: the sequence of random
prediction error covariance matrices converges weakly to a unique
invariant distribution, whose support exhibits fractal behavior. For stabilizable and detectable systems, stochastic boundedness (and hence weak convergence) holds for any non-zero observation packet arrival probability and, in particular, we can establish weak convergence at operating arrival rates well below the critical probability for mean stability (the resulting invariant measure in that situation does not possess a first moment.)
We apply the weak-Feller property
of the Markov process governing the RRE to characterize the
support of the limiting invariant distribution as the topological
closure of a countable set of points, which, in general, is not
dense in the set of positive semi-definite matrices. We use the
explicit characterization of the support of the invariant
distribution and the almost sure ergodicity of the sample paths to
easily compute statistics of the invariant distribution. A one
dimensional example illustrates that the support is a fractured
subset of the non-negative reals with self-similarity properties.
\end{abstract}

\section{Introduction}
\label{introduction}

\subsection{Background and Motivation}
\label{backmot} Named after Count Jacopo Francesco Riccati, the
man who conceived and studied it first, the Riccati equation has
received great interest in science and engineering. In particular,
its applications to control theory are widespread, ranging from
optimal to robust and stochastic control. In  Kalman
filtering, \cite{kalman}, the Riccati equation describes the evolution of the
state error covariance for linear Gaussian systems.
 We focus in the paper on the discrete
version of the Riccati equation. Kalman showed that, for a linear
time-invariant system, under detectability conditions, the Riccati
equation converges to a fixed point, which is also unique if
certain stabilizability conditions are satisfied. The result is
very powerful as it asserts that the estimation steady state error
is constant. As a consequence, the steady state estimator gain is
also constant, providing a very practical result for
implementation. The problem is more involved when the system
matrices are time-varying, and it is further complicated if, in
addition, they are random.

Initial study of Random Riccati Equations (RRE)\footnote{In the sequel the term RRE refers to the discrete time Random Riccati Equations considered in this paper.} was
motivated by the Linear Quadratic Gaussian (LQG) in optimal
control when the system parameters are random. This  leads
naturally to a RRE. In adaptive control, where the parameters of
the system are unknown and need to be identified, RRE also
arises. Initial studies of RRE are in~\cite{ckvs89}, where the
authors consider linear stochastic systems with additive white
Gaussian noise, with the added generality that the system matrices
are random and adapted to the observation process. The paper shows
that the sufficient conditions for the Kalman Filter to provide
mean and covariance of the conditionally Gaussian state estimate
are that the random matrices are finite with probability one at
each time instant. This result applies to control problems of a
linear stochastic system in the case its parameters need to be
identified recursively. More recently, Wang and Guo~\cite{wg99}
provide sufficient conditions on the stochastic Grammian to
guarantee stability of RREs.

In the past few years, RRE has
received renewed interest in the area of networked control
systems. This is concerned with
estimation and control where components, namely, sensors,
controllers, and actuators are connected via general purpose
communication channels, such as ethernet, W-LANs, or Personal area
networks (PANs), e.g., IEEE 802.15.4-based networks. In this realm,
the stochastic characteristics of the channels introduce
additional sources of randomness, non Gaussian, in the control
problem. Special interest has been given to analog erasure
channels. Under this model, the observation packet is either
dropped with probability $\bar{\gamma}$, or reaches the receiver
with probability $1-\bar{\gamma}$.
One limitation of this model is that it does not take into
account quantization. Limits of control in the presence of
quantization have been addressed
in~\cite{tm04,tm04_2,mfdn09,ms04,mde06}. Fundamental results show
that systems can be stabilized with quantization level easily
achievable by common off-the-shelf A/D converters. This makes the
infinite precision assumption realistic. The results provided in this paper therefore neglect quantization effects.

In Matveev and Savkin~\cite{ms03}, the authors consider Kalman
filtering where observations can suffer bounded delay in
communication between the sensors and the estimator.  Sinopoli et
al.~\cite{Bruno}  consider a discrete-time system in which the
arrival of an observation at the estimator is modeled as a
Bernoulli i.i.d.~random process $\gamma_t$. The observation is
received by the estimator with probability $\bar{\gamma}$. They
show that under this model the Kalman Filter is still the optimal
estimator and study the time evolution of the error covariance.
Differently from the standard Kalman Filter, the error covariance
is now a random matrix, depending on the realization of the
process $\{\gamma_{t}\}$. This is described by a RRE. They study
the asymptotic behavior (in time) of its mean to determine
stability of the filter and show that, depending on the
eigenvalues of the matrix and on the structure of the matrix,
there exists a critical value $\overline{\gamma}^{\mbox{\scriptsize{bim}}}$, such that,
if the probability of arrival of an observation at time t is
$\bar{\gamma}>\overline{\gamma}^{\mbox{\scriptsize{bim}}}$, then the expectation of the
estimation error covariance is always finite (under
stabilizability and detectability hypotheses.)  The authors
provide upper and lower bounds for this critical probability and
compute it in closed form for a few special cases. Subsequent work
~\cite{ms08} characterizes the critical value for a large class of
linear systems, showing a direct relationship with the spectral
radius of the dynamic matrix $A$.

The model proposed in~\cite{Bruno} has been widely adopted and
extended by several
authors~\cite{Liu:04,Gupta:05,Xu:05,Minyi:07,kp-fb:07j,Craig:07,xx07}.
Although many present extensions to general Markov chains and
account for smart sensors sending local estimates instead of
observations, all the results are established with respect to mean
stability, i.e., boundedness of the mean covariance. This metric
is unsatisfactory in many applications, as it does not provide
information about the fluctuations of the error covariance that
could grow and be unusable for long time intervals. We would like
to characterize the asymptotic behavior of its distribution--the
goal of this paper.

In this work, we characterize the asymptotics
 of the state error covariance for a linear Gaussian
system where observations are lost according to a Bernoulli
process, as in~\cite{Bruno}. Based on
stochastic boundedness (see Subsection~\ref{def_crit}) of the
sequence of random prediction error covariance matrices, we provide a sufficient condition (which is also necessary under broad assumptions, including stabilizability and detectability of the system in question) for the existence and uniqueness of an attracting invariant (stationary) distribution for the RRE. We show that stochastic
boundedness implies weak convergence of the sequence of random
prediction error covariance matrices to a unique invariant
distribution, irrespective of the initial condition. We show that
the mean stability considered in~\cite{Bruno} implies stochastic
boundedness and hence it is possible to operate at packet arrival
probabilities below the threshold for mean stability and converge
to an invariant distribution. In particular, for stabilizable and detectable systems, stochastic boundedness is ensured by operating at any non-zero packet arrival probability leading to weak convergence, whereas, the critical probability for boundedness in mean can be very high, depending on the instability of the system. However, operating above the
critical probability for mean stability ensures that the invariant
distribution has a finite mean, which may not hold if operated
below. Our approach is based on modelling the RRE as an
order-preserving random dynamical system (RDS)
(see~\cite{ArnoldChueshov}), possessing the property of strong
sublinearity (to be explained later.) We use a limit-set dichotomy
result for such order-preserving, strongly sublinear RDS to
establish asymptotic properties of the RRE concerning existence
and uniqueness of invariant distributions. We contrast our work with Vakili and
Hassibi~\cite{vh08} and Censi~\cite{censi09}. In~\cite{vh08}, the authors take a completely different and very interesting
approach. They use the Stieltjes transform to compute a fixed
point for the RRE associated with intermittent loss of
observation due to a Bernoulli process. Although this is numerically sound, it assumes the existence of a stationary
distribution for the error covariance, and it is applicable only to
large matrices, i.e., as $N$ tends to infinity, which are also
asymptotically free~\cite{voiculescu}.
 When the first draft of our paper was complete, we came across~\cite{censi09},
which studies weak convergence of the RRE using the theory of Iterated Function Systems (IFS) (e.g., \cite{DiaconisFreedman}.) When the
system matrix~$A$ is invertible and a non-overlapping condition is
satisfied, the RRE satisfies a mean
contraction property, leading to existence and uniqueness of an
attracting invariant distribution (see~\cite{DiaconisFreedman}). Reference~\cite{censi09} uses these results to show weak convergence of the RRE to a unique invariant distribution if the system is operated above the
critical probability for mean stability and the resulting
invariant distribution has fractal support.
By contrast, our paper shows weak convergence to an attracting invariant
distribution for the general case and at operating
points below the critical probability for mean stability.

The weak-Feller property of the Markov process governing the RRE
enables us to explicitly characterize its support of the resulting
invariant distribution. We show that its support is the topological closure (in the metric space of
positive semidefinite matrices) of a countable set of points
(given explicitly as functionals of the
deterministic fixed point of the corresponding algebraic Riccati
equation.) The above set of points is not, in
general, a dense subset of the set of positive semidefinite
matrices. A detailed study of a scalar example shows that the
support is a highly fractured subset of the non-negative reals
with self-similarity properties, thus exhibiting the
characteristics of a fractal set. Finally, the explicit
identification of the support of the invariant distribution in the
general case and almost sure
(a.s.) ergodicity of the sample paths enable us to easily compute numerically the moments (and probabilities) of the invariant distribution. In this context, we note that a complete analytic characterization of the resulting invariant measures (for example, probabilities of large excursions under the invariant measures) has been addressed more recently in the follow-up paper (\cite{Riccati-moddev}), which characterizes moderate deviations properties of the invariant measures as the packet arrival probability $\overline{\gamma}$ approaches 1.


The paper is organized as follows.
Subsection~\ref{notprel} sets notation and summarizes preliminary
results. Section~\ref{prob_form} presents
a rigorous formulation of the weak convergence problem and the
main results of the paper are stated in Section~\ref{main_res}.
The RDS formulation of the RRE is carried out in
Section~\ref{RDS_form}, while Section~\ref{sec:prop_RDS} establishes
various properties of the RRE in the context of RDS theory. The
proofs of the main results are presented in
Section~\ref{main_proof}. Subsection~\ref{scal_num} analyzes a
scalar example in detail, while numerical studies on the invariant
distribution for the general case are presented in
Subsection~\ref{num_stud}. Finally Section~\ref{conclusion}
concludes the paper.

\subsection{Notation and Preliminaries}
\label{notprel}
Denote by: $\mathbb{R}$, the reals; $\mathbb{R}^{M}$,
the $M$-dimensional Euclidean space; $\mathbb{T}$, the integers; $\mathbb{T}_{+}$, the non-negative integers; $\mathbb{N}$, the natural numbers; and $\mathcal{X}$, a generic space.
For a subset $B\subset \mathcal{X}$,  $\mathbb{I}_{B}:\mathcal{X}\longmapsto\{0,1\}$ is the indicator
function, which is~$1$ when the argument is in~$B$ and zero otherwise; and $id_{\mathcal{X}}$ is the identity function on
$\mathcal{X}$.


\textbf{Cones in partially ordered Banach spaces.}
 We summarize facts and definitions on the structure of
cones in partially ordered Banach spaces. Let $V$ be a Banach
space (over the field of the reals) with a closed (w.r.t. the
Banach space norm) convex cone $V_{+}$ and assume $V_{+}\cap
(-V_{+})=\{0\}$. The cone $V_{+}$ induces a partial order in $V$,
namely, for $X,Y\in V$, we write $X\preceq Y$, if $Y-X\in V_{+}$.
In case $X\preceq Y$ and $X\neq Y$, we write $X\prec Y$. The cone
$V_{+}$ is called solid, if it has a non-empty interior
$\mbox{int}\,V_{+}$; in that case, $V_{+}$ defines a strong
ordering in $V$, and we write $X\ll Y$, if
$Y-X\in\mbox{int}\,V_{+}$. The cone $V_{+}$ is normal if the norm
$\|\cdot\|$ of $V$ is semi-monotone, i.e., $\exists\, c>0$, s.t.
$0\preceq X\preceq Y \Rightarrow \|X\|\leq c\|Y\|$. There are
various equivalent characterizations of normality, of which we
note that the normality of $V_{+}$ ensures that the topology in
$V$ induced by the Banach space norm is compatible with the
ordering induced by $V_{+}$, in the sense that any norm-bounded
set $B\subset V$ is contained in a conic interval of the form
$[X,Y]$, where $X,Y\in V$. Finally, a cone is said to be
minihedral, if every order-bounded (both upper and lower bounded)
finite set $B\subset V$ has a supremum (here bounds are w.r.t. the
partial order.)

We focus on the separable Banach
space of symmetric $N\times N$ matrices, $\mathbb{S}^{N}$,
equipped with the induced 2-norm. The subset $\mathbb{S}^{N}_{+}$
of positive semidefinite matrices is a closed, convex, solid,
normal, minihedral cone in $\mathbb{S}^{N}$, with non-empty
interior $\mathbb{S}^{N}_{++}$, the set of positive definite
matrices. The conventions above denote the
partial and strong ordering in $\mathbb{S}^{N}$ induced by
$\mathbb{S}_{+}^{N}$. For example, we use the notation $X\gg 0$ to denote that the matrix $X\in\mathbb{S}^{N}$ is positive definite, whereas $X\succeq 0$ denotes positive semidefiniteness and $X\succ 0$ indicates that $X$ is positive semidefinite and different from the zero matrix.

\textbf{Operator theoretic preliminaries.}
 We review operator-theoretic concepts needed to
analyze the Markov process generated by the random covariance
equations, details in, for
example, \cite{Zaharopol}. Let: $(\mathcal{X},d)$ a locally
compact separable metric space $\mathcal{X}$ with
metric~$d$;  $\mathbb{B}(\mathcal{X})$ its Borel algebra;
 $B(\mathcal{X})$ the Banach space of real-valued bounded
functions on $\mathcal{X}$, equipped with the sup-norm, i.e.,
$f\in B(\mathcal{X}), \:\|f\|=\sup_{x\in\mathcal{X}}|f(x)|$; and $C_{b}(\mathcal{X})$  the subspace of $B(\mathcal{X})$ of continuous
functions.

Let $\mathcal{M}(\mathcal{X})$ be the space of finite Borel
measures on $\mathcal{X}$. It is a Banach space under the total
variation norm (see~\cite{Zaharopol} for details.) For
$\mu\in\mathcal{M}(\mathcal{X})$, we define the support of $\mu$,
$\mbox{supp}(\mu)$, by
\begin{equation}
\label{def_support}
\mbox{supp}(\mu)=\left\{x\in\mathcal{X}\left|\right.\mu(B_{\varepsilon}(x))>0,
\:\:\forall\varepsilon>0\right\}
\end{equation}
where $B_{\varepsilon}(x)$ is the open ball of radius
$\varepsilon$ centered at $x$. It follows that
$\mbox{supp}(\mu)$ is a closed set.
 An element $\mu\in\mathcal{M}(\mathcal{X})$ is positive,
i.e., $\mu\geq 0$, if
\begin{equation}
\label{MarFel1} \mu(A)\geq 0, \forall~A\in\mathbb{B}(\mathcal{X})
\end{equation}

For $f\in B(\mathcal{X})$, $\mu\in\mathcal{M}(\mathcal{X})$, we define the bilinear form
$<\cdot,\cdot>:B(\mathcal{X})\times\mathcal{M}(\mathcal{X})\longmapsto\mathbb{R}$
\begin{equation}
\label{MarFel2} <f,\mu>=\int_{\mathcal{X}}f(x)\mu(dx)
\end{equation}

A linear operator
$T:\mathcal{M}(\mathcal{X})\longmapsto\mathcal{M}(\mathcal{X})$ is
positive if $T\mu\geq 0$ for $\mu\geq 0$. It can be shown
that such positive operators are necessarily bounded.

A linear operator
$T:\mathcal{M}(\mathcal{X})\longmapsto\mathcal{M}(\mathcal{X})$ is
a contraction if $\|T\|\leq 1$.

A positive contraction
$T:\mathcal{M}(\mathcal{X})\longmapsto\mathcal{M}(\mathcal{X})$ is
a Markov operator if $\|T\mu\|=\|\mu\|, \forall \mu\geq
0$.

\begin{definition}[Markov-Feller Operator] Consider the linear operator $L:C_{b}(\mathcal{X})\longmapsto C_{b}(\mathcal{X})$ and the Markov operator
$T:\mathcal{M}(\mathcal{X})\longmapsto\mathcal{M}(\mathcal{X})$. The pair $(L,T)$ is a Markov-Feller pair if
\begin{equation}
\label{MarFel3}
<Lf,\mu>=<f,T\mu>, \forall f\in
C_{b}(\mathcal{X}),\mu\in\mathcal{M}(\mathcal{X})
\end{equation}
A Markov operator
$T:\mathcal{M}(\mathcal{X})\longmapsto\mathcal{M}(\mathcal{X})$ is
a Markov-Feller operator if there exists a linear operator
$L:C_{b}(\mathcal{X})\longmapsto C_{b}(\mathcal{X})$ such that
$(L,T)$ is a Markov-Feller pair.
\end{definition}

\textbf{Weak Convergence and Invariant probabilities.}
  Assume $(\mathcal{X},d)$ is a locally compact
separable metric space. Let $\mathcal{P}(\mathcal{X})$ be the
subset of probability measures in $\mathcal{M}(\mathcal{X})$. The sequence $\{\mu_{t}\}_{t\in\mathbb{T}_{+}}$ in
$\mathcal{P}(\mathcal{X})$ converges weakly to $\mu\in \mathcal{P}(\mathcal{X})$ if
\begin{equation}
\label{WC1a}
\lim_{t\rightarrow\infty}<f,\mu_{t}>\,=\,<f,\mu>,~~\forall~f\in
C_{b}(\mathcal{X})
\end{equation}
Weak convergence is denoted by $\mu_{t}\Longrightarrow\mu$ and is
also referred to as convergence in distribution. The weak topology
on $\mathcal{P}(\mathcal{X})$ generated by weak convergence can be metrized. In particular, e.g., \cite{Jacod-Shiryaev}, one has the
Prohorov metric $d_{p}$ on $\mathcal{P}(\mathcal{X})$, such that
the metric space $\left(\mathcal{P}(\mathcal{X}),d_{p}\right)$ is
complete, separable, and a sequence $\{\mu_{t}\}_{t\in\mathbb{T}_{+}}$ in $\mathcal{P}(\mathcal{X})$
converges weakly to $\mu$ in $\mathcal{P}(\mathcal{X})$ \emph{iff}
\begin{equation}
\label{WC1} \lim_{t\rightarrow\infty}d_{p}(\mu_{t},\mu)=0
\end{equation}

Let $(L,T)$ be a Markov-Feller pair on $(\mathcal{X},d)$. A
probability measure $\mu\in\mathcal{P}(\mathcal{X})$ is
an invariant probability for $T$ if $T\mu=\mu$. The operator~$T$
is uniquely ergodic if $T$ has exactly one invariant
probability. A probability measure $\mu^{\ast}$ is an
attracting probability for $T$ if, for any
$\mu\in\mathcal{P}(\mathcal{X})$, the sequence
$\{T^{t}\mu\}_{t\in\mathbb{T}_{+}}$ converges weakly to
$\mu^{\ast}$. In other words,
\begin{equation}
\label{WC2} \lim_{t\rightarrow\infty}<f,T^{t}\mu>~=~<f,\mu^{\ast}>\:\:
\forall \,\, f\in C_{b}(\mathcal{X}), \:\mu\in\mathcal{P}(\mathcal{X})
\end{equation}
 It follows that, if $T$ has an
attracting probability $\mu^{\ast}$, then $T$ is uniquely ergodic
(\cite{Zaharopol}.)

\section{Problem Formulation}
\label{prob_form}

\subsection{System Model}
\label{sys_model_label}

We review the model of Kalman filtering with intermittent
observations in~\cite{Bruno}. Let
\begin{equation}
\label{sys_model} \mathbf{x}_{t+1}=A\mathbf{x}_{t}+\mathbf{w}_{t}
\end{equation}
\begin{equation}
\label{sys_model1} \mathbf{y}_{t}=C\mathbf{x}_{t}+\mathbf{v}_{t}
\end{equation}
Here $\mathbf{x}_{t}\in\mathbb{R}^{N}$ is the signal (state)
vector, $\mathbf{y}_{t}\in\mathbb{R}^{M}$ is the observation
vector, $\mathbf{w}_{t}\in\mathbb{R}^{N}$ and
$\mathbf{v}_{t}\in\mathbb{R}^{M}$ are Gaussian random vectors with
zero mean and covariance matrices $Q\succeq 0$ and $R\gg 0$,
respectively. The sequences
$\{\mathbf{w}_{t}\}_{t\in\mathbb{T}_{+}}$ and
$\{\mathbf{v}_{t}\}_{t\in\mathbb{T}_{+}}$ are uncorrelated and
mutually independent. Also, assume that the initial state
$\mathbf{x}_{0}$ is a zero-mean Gaussian vector with covariance
$P_{0}$. The m.m.s.e.~predictor $\widehat{\mathbf{x}}_{t|t-1}$ of
the signal vector $\mathbf{x}_{t}$ given the observations
$\{\mathbf{y}_{s}\}_{0\leq s< t}$ is the conditional mean. It is
recursively implemented by the Kalman filter.
 The sequence of conditional prediction error
covariances, $\left\{P_{t}\right\}_{t\in\mathbb{T}_{+}}$, is then given by
\begin{eqnarray}
\label{sys_model3}
P_{t}&=&\mathbb{E}\left[\left(\mathbf{x}_{t}-\widehat{\mathbf{x}}_{t|t-1}\right)
\left(\mathbf{x}_{t}-\widehat{\mathbf{x}}_{t|t-1}\right)^{T}\left|\right.\{\mathbf{y}(s)\}_{0\leq
s<t}\right]
\\
\label{sys_model2}
P_{t+1}&=&AP_{t}A^{T}+Q-AP_{t}C^{T}\left(CP_{t}C^{T}+R\right)^{-1}CP_{t}A^{T}
\end{eqnarray}
Under the hypothesis of stabilizability of the pair $(A,Q)$ and
detectability of the pair $(A,C)$, the deterministic sequence
$\left\{P_{t}\right\}_{t\in\mathbb{T}_{+}}$ converges to a unique
value $P^{\ast}$ (which is a fixed point of the algebraic Riccati
equation~(\ref{sys_model2})) from any initial condition $P_{0}$.

This corresponds to the classical perfect
observation scenario, where the estimator has complete knowledge
of the observation packet $\mathbf{y}_{t}$ at every time $t$. With intermittent observations, the observation packets are
dropped randomly (across the communication channel to the
estimator), and the estimator receives observations at random
times. We study the intermittent observation model considered
in~\cite{Bruno}, where the channel randomness is modeled by a
sequence $\left\{\gamma_{t}\right\}_{t\in\mathbb{T}_{+}}$ of i.i.d.~Bernoulli
random variables with mean $\overline{\gamma}$ (note,
$\overline{\gamma}$ then denotes the arrival probability.) Here,
$\gamma_{t}=1$ corresponds to the arrival of the observation packet
$\mathbf{y}_{t}$ at time $t$ to the estimator, whereas a packet
dropout corresponds to $\gamma_{t}=0$. Denote by
$\widetilde{\mathbf{y}}_{t}$ the pair
\begin{equation}
\label{sys_model5}
\widetilde{\mathbf{y}}_{t}=\left(\mathbf{y}_{t}\mathbb{I}_{(\gamma_{t}=1)},\gamma_{t}\right)
\end{equation}
Under the TCP packet acknowledgement protocol
in~\cite{Bruno} (the estimator knows at each time whether the
observation packet arrived or not), the m.m.s.e.~predictor of the
signal is given by:
\begin{equation}
\label{sys_model6}
\widehat{\mathbf{x}}_{t|t-1}=\mathbb{E}\left[\mathbf{x}_{t}\left|\right.
\left\{\widetilde{\mathbf{y}}_{s}\right\}_{0\leq
s< t}\right]
\end{equation}
A modified form of the Kalman filter giving a recursive
implementation of the estimator in eqn.~(\ref{sys_model6}) is
in~\cite{Bruno}. The sequence of
conditional prediction error covariance matrices,
$\left\{P_{t}\right\}_{t\in\mathbb{T}_{+}}$, is updated according to the
following random algebraic Riccati equation (RRE):
\begin{eqnarray}
\label{sys_model8}
P_{t}&=&\mathbb{E}\left[\left(\mathbf{x}_{t}-\widehat{\mathbf{x}}_{t|t-1}\right)
\left(\mathbf{x}_{t}-\widehat{\mathbf{x}}_{t|t-1}\right)^{T}\left|\right.\left\{\widetilde{\mathbf{y}}(s)\right\}_{0\leq
s<t}\right]\\
\label{sys_model7}
P_{t+1}&=&AP_{t}A^{T}+Q-\gamma_{t}AP_{t}C^{T}\left(CP_{t}C^{T}+R\right)^{-1}CP_{t}A^{T}
\end{eqnarray}
Unlike the classical case, the sequence
$\left\{P_{t}\right\}_{t\in\mathbb{T}_{+}}$ is now random (because of its
dependence on the random sequence
$\left\{\gamma_{t}\right\}_{t\in\mathbb{T}_{+}}$.) Thus, for each $t$,
$P_{t}$ is a random element of $S^{N}_{+}$, and we denote by
$\mathbb{\mu}_{t}^{\overline{\gamma},P_{0}}$ its distribution (the
measure it induces on $S^{N}_{+}$.) The superscripts
$\overline{\gamma}$, $P_{0}$ emphasize the dependence of
$\mathbb{\mu}_{t}^{\overline{\gamma},P_{0}}$ on the packet arrival
probability and the initial condition.

In the subsequent sections, we analyze the random sequence
$\left\{P_{t}\right\}_{t\in\mathbb{T}_{+}}$ governed by
the RRE, eqn.~(\ref{sys_model7}), and establish its
asymptotic properties including the weak convergence of the
corresponding sequence
$\left\{\mathbb{\mu}_{t}^{\overline{\gamma},P_{0}}\right\}_{t\in\mathbb{T}_{+}}$
to a unique invariant distribution
$\mathbb{\mu}^{\overline{\gamma}}$ on $S^{N}_{+}$.

Before that, we set notation. Define the functions,
$f_{i}:\mathbb{S}^{N}_{+}\longmapsto\mathbb{S}^{N}_{+},\,i=0,1$, by
\begin{equation}
\label{def_f0} f_{0}(X)=AXA^{T}+Q
\end{equation}
\begin{equation}
\label{def_f1} f_{1}(X) = AXA^{T}+Q-
AXC^{T}\left(CXC^{T}+R\right)^{-1}CXA^{T}
\end{equation}
Also, define
$f:\{0,1\}\times\mathbb{S}^{N}_{+}\longmapsto\mathbb{S}^{N}_{+}$
by
\begin{eqnarray}
\label{def_f}
f(\gamma,X)&=&\mathbb{I}_{0}(\gamma)f_{0}(X)+\mathbb{I}_{1}(\gamma)f_{1}(X)\nonumber
\\ & = & AXA^{T}+Q-
\gamma AXC^{T}\left(CXC^{T}+R\right)^{-1}CXA^{T}
\end{eqnarray}
\begin{proposition}
\label{prop1} For a fixed $\gamma\in\{0,1\}$, if $R\gg 0$, the function
$f(\gamma,X):\mathbb{S}_{+}^{N}\longmapsto\mathbb{S}_{+}^{N}$ is
continuous\footnote{As stated in Subsection IB, we may assume throughout that $\mathbb{S}^{N}$ is equipped with the induced 2-norm. However, as far as topological properties like continuity etc. are considered, the exact norm is not important as long as it makes $\mathbb{S}^{N}$ complete, because all norms on a finite dimensional linear space are equivalent, i.e., generate the same topology.}  in $X$. Also, $f(\cdot)$ is jointly measurable in
$\gamma,X$.
\end{proposition}
For a fixed $\overline{\gamma}$, define the transition probability
operator $\mathbb{Q}^{\overline{\gamma}}:\mathbb{S}^{N}_{+}
\times\mathcal{B}(\mathbb{S}^{N}_{+})\longmapsto
[0,1]$ on the locally compact separable metric space
$\mathbb{S}^{N}_{+}$ by
\begin{equation}
\label{def_trQ}
\mathbb{Q}^{\overline{\gamma}}(X,B)=(1-\overline{\gamma})\mathbb{I}_{B}
\left(f_{0}(X)\right)+\overline{\gamma}\mathbb{I}_{B}\left(f_{1}(X)\right),\,\,\forall
X\in\mathbb{S}^{N}_{+},B\in\mathcal{B}\left(\mathbb{S}^{N}_{+}\right)
\end{equation}
Now, consider the canonical path space of the sequence
$\left\{P_{t}\right\}_{t\in\mathbb{T}_{+}}$,
\begin{equation}
\label{def_canomega}
\Omega^{c}=\times_{t=1}^{\infty}\mathbb{S}^{N}_{+}
\end{equation}
and $\mathcal{F}^{c}$ be the corresponding product
$\sigma$-algebra on $\Omega^{c}$. For fixed $\overline{\gamma}$ and
$P_{0}$, denote $\mathbb{P}^{\overline{\gamma},P_{0}}$ to be the
probability measure induced on
$\left(\Omega^{c},\mathcal{F}^{c}\right)$ by
$\left\{P_{t}\right\}_{t\in\mathbb{T}_{+}}$. Then, in the
sense of distribution induced on path space, the RRE generates a
Markov process $\left\{P_{t}\right\}_{t\in\mathbb{T}_{+}}$ on
$\left(\Omega^{c},\mathcal{F}^{c},\mathbb{P}^{\overline{\gamma},P_{0}}\right)$,
such that
\begin{equation}
\label{canomega1}
\mathbb{P}^{\overline{\gamma},P_{0}}\left(P_{t+1}\in
B\left|\right.P_{t}=X\right)=Q^{\overline{\gamma}}(X,B)
\end{equation}
We denote the expectation operator associated with
$\mathbb{P}^{\overline{\gamma},P_{0}}$ by
$\mathbb{E}^{\overline{\gamma},P_{0}}$. For a fixed
$\overline{\gamma}$, the family of measures
$\left\{\mathbb{P}^{\overline{\gamma},P_{0}}\right\}_{P_{0}\in\mathbb{S}^{N}_{+}}$
on $\left(\Omega^{c},\mathcal{F}^{c}\right)$ is called a Markov
family.

Let $B\left(\mathbb{S}_{+}^{N}\right)$ be the Banach space of real-valued bounded functions on $\mathbb{S}_{+}^{N}$.
For fixed $\overline{\gamma}$, define:
{\small
\begin{eqnarray}
\label{def_Lgamma}
L^{\overline{\gamma}}:B\left(\mathbb{S}_{+}^{N}\right)\longmapsto
B\left(\mathbb{S}^{N}_{+}\right):\:\:
\left(L^{\overline{\gamma}}g\right)(X)&=&\int_{\mathbb{S}_{+}^{N}}g(Y)
\mathbb{Q}^{\overline{\gamma}}(X,dY),\:\forall
g\in B\left(\mathbb{S}_{+}^{N}\right),\,X\in\mathbb{S}_{+}^{N}\\
\label{def_Tgamma}
T^{\overline{\gamma}}:\mathcal{M}\left(\mathbb{S}_{+}^{N}\right)
\longmapsto\mathcal{M}\left(\mathbb{S}_{+}^{N}\right):\:\:
\left(T^{\overline{\gamma}}\mu\right)(B)&=&\int_{\mathbb{S}^{N}_{+}}\mathbb{Q}^{\overline{\gamma}}(Y,B)\mu(dY),\:\forall
\mu\in\mathcal{M}\left(\mathbb{S}^{N}_{+}\right),\,B\in\mathcal{B}\left(\mathbb{S}^{N}_{+}\right)
\end{eqnarray}
}
We then have the following proposition (for a proof see the Appendix):
\begin{proposition}
\label{prop_MarkFell} For every $\overline{\gamma}$,
$\left(L^{\overline{\gamma}},T^{\overline{\gamma}}\right)$ is a Markov-Feller
pair on $\mathbb{S}^{N}_{+}$.
\end{proposition}

Finally, we note, that
\begin{equation}
\label{prop_MarkFell1}
\mathbb{\mu}_{t}^{\overline{\gamma},P_{0}}=\left(T^{\overline{\gamma}}\right)^{t}\delta_{P_{0}},\: \forall
t
\end{equation}
where $\delta_{P_{0}}$ denotes the Dirac probability measure
concentrated at $P_{0}$.

\subsection{Stability notions and critical probabilities}
\label{def_crits}
\label{def_crit}


There are various stability notions for the random
sequence $\left\{P_{t}\right\}_{t\in\mathbb{T}_{+}}$. In this work, we
consider two: the first is
stochastic boundedness (uniform boundedness in probability); and the second is the more stronger  bounded
in mean stability.
\begin{definition}[Stochastic boundedness]
\label{stoch_bound_def} Consider fixed $\overline{\gamma}$ and
$P_{0}\in\mathbb{S}^{N}_{+}$. The sequence
$\left\{P_{t}\right\}_{t\in\mathbb{T}_{+}}$ is stochastically bounded
(s.b.) if
\begin{equation}
\label{stoch_bound_eqn}
\lim_{N\rightarrow\infty}\sup_{t\in\mathbb{T}_{+}}\mathbb{P}^{\overline{\gamma},
P_{0}}\left(\left\|P_{t}\right\|>N\right)=0
\end{equation}
The corresponding sequence of measures
$\left\{\mathbb{\mu}_{t}^{\overline{\gamma},P_{0}}\right\}_{t\in\mathbb{T}_{+}}$
is said to be tight (see~\cite{Jacod-Shiryaev}.)
\end{definition}
\begin{definition}[Boundedness in mean]
\label{bound_inmean_def} Consider fixed $\overline{\gamma}$ and
$P_{0}\in\mathbb{S}^{N}_{+}$. The sequence
$\left\{P_{t}\right\}_{t\in\mathbb{T}_{+}}$ is bounded in mean
(b.i.m.) if there exists $M^{\overline{\gamma},P_{0}}$, such that,
\begin{equation}
\label{bim_def}
\sup_{t\in\mathbb{T}_{+}}\mathbb{E}^{\overline{\gamma},P_{0}}\left[P_{t}\right]\preceq
M^{\overline{\gamma},P_{0}}
\end{equation}
(Note, that the supremum above is taken w.r.t.~the partial order
in $\mathbb{S}^{n}$.)
\end{definition}
We note here, that the above stability notions, applies to all systems irrespective of properties like stabilizability, detectability. Stochastic boundedness provides a trade-off between the
permissible estimation error margin and performance guarantee
uniformly over all time $t$. Stochastic boundedness is weaker than
bounded in mean stability, as indicated by the following
proposition (proof in the Appendix.)
\begin{proposition}
\label{imp_sb} Consider fixed $\overline{\gamma}$ and
$P_{0}\in\mathbb{S}^{N}_{+}$. If the sequence
$\left\{P_{t}\right\}_{t\in\mathbb{T}_{+}}$ is b.i.m., then it is s.b.
\end{proposition}
For most interesting cases, the stability of the
sequence $\left\{P_{t}\right\}_{t\in\mathbb{T}_{+}}$ depends on
$\overline{\gamma}$. Indeed, as exhibited in~\cite{Bruno}, there
exist critical probabilities marking sharp transitions in the
stability behavior. For the above stability notions,
consider the following critical probabilities:
\begin{eqnarray}
\label{def_critsb}
\overline{\gamma}^{\mbox{\scriptsize sb}}&=&\inf\left\{\overline{\gamma}\in
[0,1]:\:\left\{P_{t}\right\}_{t\in\mathbb{T}_{+}} \,\mbox{is s.b.},\,\forall
P_{0}\in\mathbb{S}^{N}_{+}\right\}\\
\label{def_critbim}
\overline{\gamma}^{\mbox{\scriptsize bim}}&=&\inf\left\{\overline{\gamma}\in
[0,1]:\:\left\{P_{t}\right\}_{t\in\mathbb{T}_{+}}\;\mbox{is b.i.m.},\,\forall
P_{0}\in\mathbb{S}^{N}_{+}\right\}
\end{eqnarray}
Thus, $\overline{\gamma}^{\mbox{\scriptsize sb}}$ marks a transition in stochastic
boundedness of the sequence $\left\{P_{t}\right\}_{t\in\mathbb{T}_{+}}$, in
the sense that, if the operating\footnote{Note for
stochastic boundedness, one has to operate strictly above
$\overline{\gamma}^{\mbox{\scriptsize sb}}$, because the infimum in
eqn.~(\ref{def_critsb}) may not be attainable.}
$\overline{\gamma}>\overline{\gamma}^{\mbox{\scriptsize sb}}$, the sequence
$\left\{P_{t}\right\}_{t\in\mathbb{T}_{+}}$ is s.b.~for all initial
conditions $P_{0}$, whereas it explodes if operated below
$\overline{\gamma}^{\mbox{\scriptsize sb}}$. A similar interpretation holds for
$\overline{\gamma}^{\mbox{\scriptsize bim}}$.

In~\cite{Bruno}, upper and lower bounds for
$\overline{\gamma}^{\mbox{\scriptsize bim}}$ were obtained. Precisely, the following was shown:
\begin{result}[\cite{Bruno}]
\label{bruno_res} For $\left(A,Q^{1/2}\right)$ stabilizable, $(A,C)$
detectable, and $A$ unstable, then $\exists\, \overline{\gamma}^{\mbox{\scriptsize bim}}\in
[0,1)$, s.t.
\begin{eqnarray}
\label{bruno_res1}
&&
\lim_{t\rightarrow\infty}\mathbb{E}^{\overline{\gamma},P_{0}}
\left[P_{t}\right]=\infty,\:\mbox{for}\,0\leq\overline{\gamma}\leq\overline{\gamma}^{\mbox{\scriptsize bim}}\,\,\mbox{and}\,\,\exists\,\,
P_{0}\succeq 0
\\
\label{bruno_res2}
&&
\left\{P_{t}\right\}_{t\in\mathbb{T}_{+}}
\,\mbox{is b.i.m.},\:\mbox{for}\,\overline{\gamma}^{\mbox{\scriptsize bim}}<\overline{\gamma}\leq
1\,\,\mbox{and}\,\,\forall\,\, P_{0}\succeq 0
\end{eqnarray}
Also, $\overline{\gamma}^{\mbox{\scriptsize bim}}_{l}\leq \overline{\gamma}^{\mbox{\scriptsize bim}} \leq \overline{\gamma}^{\mbox{\scriptsize bim}}_{u}$ where
\begin{eqnarray}
\label{bruno_res3}
\overline{\gamma}^{\mbox{\scriptsize bim}}_{l}&=&1-\frac{1}{\alpha^{2}}
\\
\label{bruno_res4}
\overline{\gamma}^{\mbox{\scriptsize bim}}_{u}&=&\inf\left\{\overline{\gamma}\in
[0,1]:\:\exists
(\widehat{K},\widehat{X}),\,\,\mbox{s.t.}\,\,\widehat{X}\gg\phi^{\overline{\gamma}}\,\,\left(\widehat{K},\widehat{X}\right)\right\}
\end{eqnarray}
where $\alpha$ is the absolute value of the largest eigenvalue of
$A$, and the operator $\phi^{\overline{\gamma}}$ is
\begin{equation}
\label{bruno_res5}
\phi^{\overline{\gamma}}\left(K,X\right)=
\left(1-\overline{\gamma}\right)\left(AXA^{T}+Q\right)+\overline{\gamma}\left(FXF^{T}+V\right)
\end{equation}
and $F=A+KC$, $V=Q+KRK^{T}$.
\end{result}
The above result provides  computable upper and lower
bounds on the critical probability $\overline{\gamma}^{\mbox{\scriptsize bim}}$.

The following proposition relates the two critical probabilities:
\begin{proposition}
\label{prop_crel} For
$\left(A,Q^{1/2}\right)$ stabilizable, $(A,C)$ detectable, and $A$
unstable, then
\begin{itemize}
\item[i)] For a general system
$\overline{\gamma}^{\mbox{\scriptsize sb}}\leq\overline{\gamma}^{\mbox{\scriptsize bim}}$.

\item[ii)] If, in addition,  $\left(A,Q^{1/2}\right)$ is stabilizable and $(A,C)$ is detectable, then $\overline{\gamma}^{\mbox{\scriptsize sb}}=0$.
\end{itemize}
\end{proposition}
\begin{proof} The proofs of part [i] and part [ii] under the additional assumption of invertible $C$ are provided in the Appendix. The proof of part [ii] for the general case of stabilizable and detectable systems can be found in the follow-up paper~\cite{Riccati-moddev}.
\end{proof}
From the above it is clear that, in general, any upper bound on
$\overline{\gamma}^{\mbox{\scriptsize bim}}$ is also an upper bound on
$\overline{\gamma}^{\mbox{\scriptsize sb}}$. Part ii) of the above proposition shows that under most reasonable assumptions
$\overline{\gamma}^{\mbox{\scriptsize sb}}=0$. For such systems, operating at any
$\overline{\gamma}>0$ guarantees stochastic boundedness\footnote{We note here that for stable systems, the infimum in the definition of $\overline{\gamma}^{\mbox{\scriptsize{sb}}}$ is attained and the sequence $\{P_{t}\}$ is stochastically bounded for $\overline{\gamma}=0$. However, for unstable systems to achieve stochastic boundedness, it is necessary to operate at $\overline{\gamma}>0$, as $\overline{\gamma}=0$ would imply that no observations arrive in the infinite horizon leading to a.s. unboundedness of the sequence $\{P_{t}\}$.}, whereas
if one needs to remain bounded in mean, one has to operate at
$\overline{\gamma}>1-\frac{1}{\alpha^{2}}$, which can be large if $A$ is highly unstable.This is
important to the system designer, because, if the design criterion
is stochastic boundedness (i.e., boundedness in probability)
rather than boundedness in mean, one may operate at a value of
$\overline{\gamma}$ strictly lower than $\overline{\gamma}^{\mbox{\scriptsize bim}}$.

In the next section, we state and discuss the main results of this paper. Among others, we show that operating
above $\overline{\gamma}^{\mbox{\scriptsize sb}}$ guarantees the existence of a
unique invariant distribution (independent of the initial
condition $P_{0}$) to which the random sequence
$\left\{P_{t}\right\}_{t\in\mathbb{T}_{+}}$ converges weakly.

\section{Main Results: Invariant Distribution}
\label{main_res} 
 The first result concerns the weak convergence properties of 
$\left\{P_{t}\right\}_{t\in\mathbb{T}_{+}}$ generated by the RRE.
\begin{theorem}
\label{main_th}Assume: $\left(A,Q^{1/2}\right)$ stabilizable;
$\left(A,C\right)$ detectable;
$Q$ positive definite; fixed 
 $\overline{\gamma}$ and
 $P_{0}\in\mathbb{S}_{+}^{N}$. Then:
\begin{itemize}
\item[i)] If $\overline{\gamma}>0$, the sequence $\left\{P_{t}\right\}_{t\in\mathbb{T}_{+}}$ is stochastically bounded and there exists a unique invariant distribution
$\mathbb{\mu}^{\overline{\gamma}}$ s.t.~the sequence
$\left\{P_{t}\right\}_{t\in\mathbb{T}_{+}}$ (or  sequence
$\left\{\mathbb{\mu}_{t}^{\overline{\gamma},P_{0}}\right\}_{t\in\mathbb{T}_{+}}$
of measures) converges weakly to
$\mathbb{\mu}^{\overline{\gamma}}$ from any initial condition
$P_{0}$.

In other words, 
 the operator $T^{\overline{\gamma}}$ is uniquely ergodic with
attracting probability $\mathbb{\mu}^{\overline{\gamma}}$.
\item[ii)] If, in addition,
$\overline{\gamma}>\overline{\gamma}^{\mbox{\scriptsize bim}}$, the corresponding
unique invariant measure $\mathbb{\mu}^{\overline{\gamma}}$ has
finite mean
\begin{equation}
\label{main_th2}
\int_{\mathbb{S}^{N}_{+}}Y\mathbb{\mu}^{\overline{\gamma}}(dY)<\infty
\end{equation}
\end{itemize}
\end{theorem}
%
 Theorem~\ref{main_th} states that, for stabilizable and detectable systems, if $\overline{\gamma}>0$, the
sequence $\left\{P_{t}\right\}_{t\in\mathbb{T}_{+}}$ converges in
distribution to a unique invariant distribution, irrespective of
the initial condition. In particular, one may operate below
$\overline{\gamma}^{\mbox{\scriptsize bim}}$ and still converge
to a unique invariant distribution. However, operating at
$0<\overline{\gamma}<\overline{\gamma}^{\mbox{\scriptsize bim}}$
may not guarantee that the corresponding invariant distribution
$\mathbb{\mu}^{\overline{\gamma}}$ has finite mean.

We discuss several implications of Theorem~\ref{main_res}. First, we note from the discussion following Proposition~\ref{prop_crel} (especially the footnote) that for stable systems, $\overline{\gamma}=0$ also leads to stochastic boundedness of the sequence $\{P_{t}\}$, which does not hold for unstable systems. Thus for stable systems, the conclusions of Theorem~\ref{main_res} hold not only for $\overline{\gamma}>0$, but alos for $\overline{\gamma}=0$.

Finally, we note that, Theorem~\ref{main_res} as stated above in the context of stabilizable and detectable systems is, in fact, more general. As can be noted from the
proof of Theorem~\ref{main_res} (Subsection~\ref{proof_main_th}), in general, a sufficient condition for the existence and uniqueness of an attracting invariant measure, is stochastic boundedness. Thus, for a general system (for which stabilizability, detectability may not be verified), operating above the critical probability $\overline{\gamma}^{\mbox{\scriptsize{sb}}}$ of stochastic boundedness is sufficient to guarantee weak convergence to a unique invariant distribution. In the case of stabilizable and detectable systems, by Proposition~\ref{prop_crel}, $\overline{\gamma}^{\mbox{\scriptsize{sb}}}=0$ and hence stochastic boundedness is ensured by operating under any $\overline{\gamma}>0$.

The second result explicitly determines the support of the
invariant measure $\mathbb{\mu}^{\overline{\gamma}}$.
\begin{theorem}
\label{supp_inv}Assume: $\left(A,Q^{1/2}\right)$ stabilizable;
$\left(A,C\right)$ detectable;
$Q$ positive definite. Define the set
$\mathcal{S}\subset\mathbb{S}^{N}_{+}$ by
\begin{equation}
\label{supp_inv1} \mathcal{S}=\left\{f_{i_{1}}\circ
f_{i_{2}}\circ\cdots\circ
f_{i_{s}}\left(P^{\ast}\right)\left|\right.i_{r}\in\{0,1\},\,1\leq r\leq
s,\:s\in\mathbb{T}_{+}\right\}
\end{equation}
where $P^{\ast}$ is the fixed point of the (deterministic) Riccati equation
Then\footnote{In eqn.~\ref{supp_inv1} $s$ can take the
value 0, implying $P^{\ast}\in\mathcal{S}$.}, if $0<\overline{\gamma}<1$,
\begin{equation}
\label{supp_inv2}
\mbox{supp}\left(\mathbb{\mu}^{\overline{\gamma}}\right)=\mbox{cl}(\mathcal{S})
\end{equation}
where $\mbox{cl}(\mathcal{S})$ denotes the topological closure of
$\mathcal{S}$ in $\mathbb{S}^{N}_{+}$.
In particular, we have
\begin{equation}
\label{main_th1}
\mathbb{\mu}^{\overline{\gamma}}
\left(\left\{Y\in\mathbb{S}^{N}_{+}\left|\right.Y\succeq
P^{\ast}\right\}\right)=1
\end{equation}
\end{theorem}
 Theorem~\ref{supp_inv} states that, for
$0<\overline{\gamma}<1$,
  $\mbox{supp}\left(\mathbb{\mu}^{\overline{\gamma}}\right)$ is independent of $\overline{\gamma}$ and is given by the closure of the countable
set $\mathcal{S}$ (but the distribution is dependent on the value
of $\overline{\gamma}$.) If $\overline{\gamma}=1$, it
reduces to the deterministic Kalman filtering, and the invariant measure is a Dirac mass at $P^{\ast}$.

The fact that the invariant measure is concentrated on the conic
interval $\left[P^{\ast},\infty\right)$, where $P^{\ast}$ is the fixed point
of the algebraic Riccati equation, eqn.~(\ref{sys_model2}), is
quite natural (but not obvious), as one cannot expect to obtain
better performance with intermittent observations.

The set $\mathcal{S}$ is not generally dense
in $\left[P^{\ast},\infty\right)$ and the support is an unbounded fractured
(many holes) subset of $\mathbb{S}^{N}_{+}$. We study a
scalar example to show, both
analytically and numerically, that the invariant measure exhibits
fractal properties, i.e., the support of the measure is a highly
fractured subset of the positive reals and exhibits
self-similarity.

The next three sections are devoted to the proofs of
Theorems~\ref{main_th},\ref{supp_inv}. The proof of
Theorem~\ref{main_th} relies on the theory of random dynamical
systems (RDS), and Theorem~\ref{supp_inv} uses the Markov-Feller
property of the transition operator. Section~\ref{RDS_form}
summarizes results on RDS and models the RRE as an RDS. Section~\ref{sec:prop_RDS} establishes
properties of the RRE as an RDS. We complete
the proof of Theorem~\ref{main_th} in Subsection~\ref{proof_main_th}, whereas
Theorem~\ref{proof_supp_inv} is proved in Subsection~\ref{proof_supp_inv}.

\section{Random Dynamical System Formulation}
\label{RDS_form} We start by defining a random dynamical system
(RDS). We follow the notation
in~\cite{ArnoldChueshov,Chueshov}.
\begin{definition}[RDS]\label{defn_RDS} A RDS with (one-sided) time
$\mathbb{T}_{+}$ and state space $\mathcal{X}$ is the pair
$(\theta,\varphi)$: 
\begin{itemize}
\item[\textbf{A)}] A metric dynamical system
$\theta=\left(\Omega,\mathcal{F},\mathbb{P},
\left\{\theta_{t},t\in\mathbb{T}\right\}\right)$
with two-sided time $\mathbb{T}$, i.e., a probability space
$(\Omega,\mathcal{F},\mathbb{P})$ with a family of transformations
$\{\theta_{t}:\Omega\longmapsto\Omega\}_{t\in\mathbb{T}}$ such
that\footnote{The function $id_{\Omega}$ denotes the identity map on $\Omega$, i.e., for all $\omega\in\Omega$, $id_{\Omega}(\omega)=\omega$.}
\begin{itemize}
\item[\textbf{A.1)}]$\theta_{0}=id_{\Omega},
\:\,\theta_{t}\circ\theta_{s}=\theta_{t+s},\:\,\forall
t,s\in\mathbb{T}$
\item[\textbf{A.2)}]$(t,\omega)\longmapsto\theta_{t}\omega$ is
measurable.
\item[\textbf{A.3)}]$\theta_{t}\mathbb{P}=\mathbb{P}\:\forall
t\in\mathbb{T}$, i.e., $\mathbb{P}\left(\theta_{t}B\right)=\mathbb{P}(B)$ for all $B\in\mathcal{F}$ and all $t\in\mathbb{T}$.
\end{itemize}
\item[\textbf{B)}] A cocycle $\varphi$ over $\theta$ of continuous
mappings of $\mathcal{X}$ with time $\mathbb{T}_{+}$, i.e., a
measurable mapping
\begin{equation}
\label{def_RDS} \varphi:\mathbb{T}_{+}\times\Omega\times
\mathcal{X}\rightarrow \mathcal{X}, \:(t,\omega,X)\longmapsto\varphi(t,\omega,X)
\end{equation}
\begin{itemize}
\item[\textbf{B.1)}] The mapping
$X\longmapsto\varphi(t,\omega,X)\equiv\varphi(t,\omega)X$ is
continuous in $X$  $\forall\,t\in\mathbb{T}_{+},\,\omega\in\Omega$.
\item[\textbf{B.2)}] The mappings
$\varphi(t,\omega)\doteq\varphi(t,\omega,\cdot)$ satisfy the
cocycle property:
\begin{equation}
\label{def_RDS1}
\varphi(0,\omega)=id_{\mathcal{X}},\:\:\varphi(t+s,\omega)=
\varphi(t,\theta_{s}\omega)\circ\varphi(s,\omega),\:\:\forall\,
t,s\in\mathbb{T}_{+},\,\,\omega\in\Omega
\end{equation}
\end{itemize}
\end{itemize}
\end{definition}
 In a RDS, randomness is captured by the  space
$(\Omega,\mathcal{F},\mathbb{P})$. Iterates indexed by $\omega$
indicate pathwise construction. For example, if $X_{0}$ is the
deterministic  state 
 at $t=0$, the state at $t\in\mathbb{T}_{+}$ is 
\begin{equation}
\label{def_RDS2} X_{t}(\omega)=\varphi\left(t,\omega,X_{0}\right)
\end{equation}
The measurability assumptions 
  guarantee that the state $X_{t}$ is a well-defined random variable.
Also,  the iterates are defined for non-negative
(one-sided) time; however, the family of transformations
$\left\{\theta_{t}\right\}$ is two-sided, which is purely for technical
convenience, as will be seen later.

We now show that the sequence $\left\{P_{t}\right\}$ generated by the RRE can be modeled as the sequence of iterates (in
the sense of distributional equivalence) of a suitably defined
RDS.

Fix $\overline{\gamma}$ and define:
$\left(\widetilde{\Omega},\widetilde{\mathcal{F}},
\widetilde{\mathbb{P}}^{\overline{\gamma}}\right)$,
where $\widetilde{\Omega}=\{0,1\}$,
$\widetilde{\mathcal{F}}=2^{\{0,1\}}$ and
$\widetilde{\mathbb{P}}^{\overline{\gamma}}(\{1\})=\overline{\gamma}$; and
 the product space,
$\left(\Omega,\mathcal{F},\mathbb{P}^{\overline{\gamma}}\right)$, where
$\Omega=\times_{t\in\mathbb{T}}\widetilde{\Omega}$ and
$\mathcal{F}$ and $\mathbb{P}^{\overline{\gamma}}$ are the
 product $\sigma$-algebra and the product
measure\footnote{Note the difference between the measure
$\mathbb{P}^{\overline{\gamma}}$ and the measures
$\mathbb{P}^{\overline{\gamma},P_{0}}$ defined in
Subsection~\ref{sys_model_label}.}. From the
construction, a sample point $\omega\in\Omega$ is a
two-sided binary sequence and, since
$\mathbb{P}^{\overline{\gamma}}$ is the product of
$\widetilde{\mathbb{P}}^{\overline{\gamma}}$, the projections are
i.i.d.~binary random variables with probability of one being
$\overline{\gamma}$. Define the family of transformations
$\left\{\theta^{R}_{t}\right\}_{t\in\mathbb{T}}$ on $\Omega$ as the family of
left-shifts
\begin{equation}
\label{def_RDS3} \theta_{t}^{R}\omega(\cdot)=\omega(t+\cdot),\:\forall
t\in\mathbb{T}
\end{equation}
With this, the space
$\left(\Omega,\mathcal{F},\mathbb{P}^{\overline{\gamma}},
\left\{\theta_{t}^{R},t\in\mathbb{T}\right\}\right)$
is the canonical path space of a two-sided stationary (in
fact, i.i.d.) sequence equipped with the left-shift operator;
hence, (e.g., \cite{Kallenberg}) it satisfies the
Assumptions~\textbf{A.1)-A.3)} to be a metric dynamical system;
in fact, it is ergodic.

Recall the Riccati iterates $f(\gamma,X)$ in eqn.~(\ref{def_f}).
Define the function
$\widetilde{f}:\Omega\times\mathbb{S}_{+}^{N}\longmapsto\mathbb{S}_{+}^{N}$
by
\begin{equation}
\label{def_RDS4} \widetilde{f}(\omega,X)=f(\omega(0),X)
\end{equation}
 Since the
projection map from $\omega$ to $\omega(0)$ is measurable
(continuous) and $f(\cdot)$ is jointly measurable in $\gamma,X$
(Proposition~\ref{prop1}), it follows that $\widetilde{f}(\cdot)$ is
jointly measurable in $\omega,X$. Define the function
$\varphi^{R}:\mathbb{T}_{+}\times\Omega\times\mathbb{S}_{+}
\longmapsto\mathbb{S}_{+}$
by
\begin{eqnarray}
\label{def_RDS5} \varphi^{R}(0,\omega,X)&=&X,\:\hspace{4.05cm}\forall\,\,\omega,X\\
\label{RDS6}
\varphi^{R}(1,\omega,X)&=&\widetilde{f}(\omega,X),\:\hspace{3.1cm}\forall\,\,\omega,X\\
\label{def_RDS7}
\varphi^{R}(t,\omega,X)&=&\widetilde{f}
\left(\theta^{R}_{t-1}\omega,\varphi^{R}(t-1,\omega,X)\right),\:\forall\,\,
t>1,\omega,X
\end{eqnarray}
It follows from the measurability of the transformations
$\left\{\theta_{t}^{R}\right\}$, the measurability of $\widetilde{f}(\cdot)$,
and the fact that $\mathbb{T}_{+}$ is countable that the function
$\varphi^{R}(t,\omega,X)$ is jointly measurable in $t,\omega,X$.
Finally, $\varphi^{R}(\cdot)$ defined above satisfies
Assumption~\textbf{B.2)} by virtue of Proposition~\ref{prop1}, and
Assumption~\textbf{B.3)} follows by the construction given by
eqns.~(\ref{def_RDS5}-\ref{def_RDS7}). Thus, the pair
$\left(\theta^{R},\varphi^{R}\right)$ is an RDS over $\mathbb{S}_{+}^{N}$.
Given a deterministic initial condition
$P_{0}\in\mathbb{S}_{+}^{N}$, it follows that the sequence
$\left\{P_{t}\right\}_{t\in\mathbb{T}_{+}}$ generated by the RRE
eqn.~(\ref{sys_model7}) is equivalent in the sense of distribution
to the sequence
$\left\{\varphi^{R}(t,\omega,P_{0})\right\}_{t\in\mathbb{T}_{+}}$ generated
by the iterates of the above constructed RDS, i.e.,
\begin{equation}
\label{def_RDS100}
P_{t}\ndtstile{}{d}\varphi^{R}(t,\omega,P_{0}),\:\forall
t\in\mathbb{T}_{+}
\end{equation}
Indeed, by studying eqn.~(\ref{def_RDS7}), we note that the iterate $\varphi^{R}(t,\omega,P_{0})$ at time $t$ is obtained by applying the map $f_{\omega_{t-1}}$ to $\varphi^{R}(t-1,\omega,P_{0})$ and by construction, the random variable $\omega_{t-1}$ is 1 with probability $\overline{\gamma}$ and 0 with probability $1-\overline{\gamma}$.
Thus, investigating the distributional properties of
$\left\{P_{t}\right\}_{t\in\mathbb{T}_{+}}$ is equivalent to analyzing the
distributional properties of
$\left\{\varphi^{R}(t,\omega,P_{0})\right\}_{t\in\mathbb{T}_{+}}$, which we
carry out in the rest of the paper.

In the sequel, we use the pair $(\theta,\varphi)$ to denote a
generic RDS and $\left(\theta^{R},\varphi^{R}\right)$ for the one constructed
above for the RRE.

\section{Properties of the RDS $(\theta,\varphi)$}
\label{sec:prop_RDS}

\subsection{Facts about generic RDS}
\label{facts_genRDS} We review concepts on RDS
(see~\cite{ArnoldChueshov,Chueshov} for details.)
 Consider a generic RDS $(\theta,\varphi)$ with state space
$\mathcal{X}$ as in Definition~\ref{defn_RDS}. Assume that $\mathcal{X}$ is a non-empty subset of a real Banach
space $V$ with a closed, convex, solid, normal (w.r.t.~the Banach
space norm), minihedral cone $V_{+}$. Denote by $\preceq$ the
partial order induced by $V_{+}$ in $\mathcal{X}$ and $<<$ denotes
the corresponding strong order. Although the development that
follows may hold for arbitrary $\mathcal{X}\subset V$, in the
sequel, we assume $\mathcal{X}=V_{+}$ (which is true for the RDS
$\left(\theta^{R},\varphi^{R}\right)$ modeling the RRE.)

\begin{definition}[Order-Preserving RDS]
\label{Order} A RDS $(\theta,\varphi)$ with state space $V_{+}$
is order-preserving if
\begin{equation}
X\preceq
Y\:\Longrightarrow\:\varphi(t,\omega,X)\preceq\varphi(t,\omega,Y),\:\forall
t\in\mathbb{T}_{+},\:\omega\in\Omega,\:X,Y\in V_{+}
\end{equation}
\end{definition}

\begin{definition}[Sublinearity]\label{sublinearity} An order-preserving RDS $(\theta,\varphi)$ with state space $V_{+}$ is
sublinear if for every $X\in V_{+}$ and $\lambda\in (0,1)$
we have
\begin{equation}
\label{prop_RDS}
\lambda\varphi(t,\omega,X)\preceq\varphi(t,\omega,\lambda
X),\:\forall t>0,\:\omega\in\Omega
\end{equation}
The RDS is strongly sublinear if in addition to
eqn.~(\ref{prop_RDS}), we have
\begin{equation}
\label{prop_RDS2}
\lambda\varphi(t,\omega,X)\ll\varphi(t,\omega,\lambda X),\:\forall
t>0,\:\omega\in\Omega,\:X\in\mbox{int}\,V_{+}
\end{equation}
\end{definition}

\begin{definition}[Equilibrium]\label{equilibrium} A random
variable $u:\Omega\longmapsto V_{+}$ is called an equilibrium
(fixed point, stationary solution) of the RDS $(\theta,\varphi)$
if it is invariant under $\varphi$, i.e.,
\begin{equation}
\label{prop_RDS3}
\varphi\left(t,\omega,u(\omega)\right)=u\left(\theta_{t}\omega\right),\:\forall
t\in\mathbb{T}_{+},\:\omega\in\Omega
\end{equation}
If eqn.~(\ref{prop_RDS3}) holds $\forall\,\omega\in\Omega$,
except on  set of $\mathbb{P}$ measure zero, $u$ is an
almost equilibrium.
\end{definition}
Since the transformations $\left\{\theta_{t}\right\}$ are
measure-preserving, i.e.,
$\theta_{t}\mathbb{P}=\mathbb{P},\:\forall t$, we have
\begin{equation}
\label{prop_RDS4} u\left(\theta_{t}\omega\right)\ndtstile{}{d}
u(\omega),\:\forall t
\end{equation}
By eqn.~(\ref{prop_RDS3}), for an
almost equilibrium $u$, the iterates in the sequence
$\left\{\varphi\left(t,\omega,u(\omega)\right)\right\}_{t\in\mathbb{T}_{+}}$
have the same distribution, which is the distribution of $u$.

%

\begin{definition}[Orbit]\label{orbit} For a random variable $u:\Omega\longmapsto
V_{+}$, we define the \emph{forward} orbit $\eta^{f}_{u}(\omega)$
emanating from $u(\omega)$ as the random set
$\left\{\varphi\left(t,\omega,u(\omega)\right)\right\}_{t\in\mathbb{T}_{+}}$.
The forward orbit gives the sequence of iterates of the RDS
starting at $u$.

Although $\eta^{f}_{u}$ is the object of  interest,
for technical convenience (as will be seen later), we also define the
\emph{pull-back} orbit $\eta^{b}_{u}(\omega)$ emanating from $u$
as the random set
$\left\{\varphi\left(t,\theta_{-t}\omega,u\left(\theta_{-t}\omega\right)
\right)\right\}_{t\in\mathbb{T}_{+}}$.
\end{definition}
We establish asymptotic properties for
the pull-back orbit $\eta^{b}_{u}$. This is because it is more convenient and because analyzing $\eta_{u}^{b}$ leads to understanding the asymptotic distributional properties for $\eta^{f}_{u}$. In fact, the random sequences
$\left\{\varphi\left(t,\omega,u(\omega)\right)\right\}_{t\in\mathbb{T}_{+}}$
and $\left\{\varphi\left(t,\theta_{-t}\omega,u\left(\theta_{-t}\omega\right)
\right)\right\}_{t\in\mathbb{T}_{+}}$ are equivalent in distribution. In other words,
\begin{equation}
\label{orbit1}
\varphi\left(t,\omega,u(\omega)\right)\ndtstile{}{d}\varphi
\left(t,\theta_{-t}\omega,u\left(\theta_{-t}\omega\right)\right),\:\forall
t\in\mathbb{T}_{+}
\end{equation}
This follows from
$\theta_{t}\mathbb{P}=\mathbb{P},\:\forall t\in\mathbb{T}$ (hence the random objects $\omega$ and $\theta_{t}\omega$ possess the same distribution.) Thus,
in particular, we have the following assertion.
\begin{lemma}
\label{orbit_lemma} Let the sequence
$\left\{\varphi\left(t,\theta_{-t}\omega,u\left(\theta_{-t}\omega\right)
\right)\right\}_{t\in\mathbb{T}_{+}}$
converge in distribution to a measure $\mu$ on $V_{+}$, where
$u:\Omega\longmapsto V_{+}$ is a random variable. Then the
sequence
$\left\{\varphi\left(t,\omega,u(\omega)\right)\right\}_{t\in\mathbb{T}_{+}}$
also converges in distribution to the measure $\mu$.
\end{lemma}

We now introduce some notions of boundedness of RDS, which will be
used in the sequel.

\begin{definition}[Boundedness]\label{boundedness} Let
$a:\Omega\longmapsto V_{+}$ be a random variable. The pull-back
orbit $\eta_{a}^{b}(\omega)$ emanating from $a$ is
bounded on $U\in\mathcal{F}$ if there exists a random variable $C$
on $U$ s.t.
\begin{equation}
\label{boundedness1}
\left\|\varphi\left(t,\theta_{-t}\omega,a\left(\theta_{-t}\omega\right)
\right)\right\|\leq
C(\omega),\:\forall t\in\mathbb{T}_{+},~\omega\in U
\end{equation}
\end{definition}

\begin{definition}[Conditionally Compact RDS]\label{cond_comp} An
RDS $(\theta,\varphi)$ in $V_{+}$ is conditionally
compact if for any $U\in\mathcal{F}$ and pull-back orbit
$\eta_{a}^{b}(\omega)$ that is bounded on $U$ there exists a
family of compact sets $\{K(\omega)\}_{\omega\in U}$ s.t.
\begin{equation}
\label{cond_comp1}
\lim_{t\rightarrow\infty}\mbox{dist}\left(\varphi\left(t,\theta_{-t}\omega,
a\left(\theta_{-t}\omega\right)\right),K(\omega)\right)=0,\:\omega\in
U
\end{equation}
\end{definition}
It is to be noted that conditionally compact is a topological
property of the space $V_{+}$. In particular, an RDS in a finite
dimensional space $V_{+}$ is conditionally compact.

We now state a limit set dichotomy result for a class of
sublinear, order-preserving RDS.
\begin{theorem}[Corollary 4.3.1. in~\cite{Chueshov}]\label{LSD}
Let $V$ be a separable Banach space with a normal solid cone
$V_{+}$. Assume that $(\theta,\varphi)$ is a strongly sublinear
conditionally compact order-preserving RDS over an ergodic metric
dynamical system $\theta$. Suppose that $\varphi(t,\omega,0)\gg 0$
for all $t>0$ and $\omega\in\Omega$. Then precisely one of the
following applies:
\begin{itemize}
\item[\textbf{(a)}] For any $X\in V_{+}$ we have
\begin{equation}
\label{LSD1}
\mathbb{P}\left(\lim_{t\rightarrow\infty}\left\|\varphi\left(t,\theta_{-t}\omega,X\right)\right\|=\infty\right)=1
\end{equation}
\item[\textbf{(b)}] There exists a unique almost equilibrium
$u(\omega)\gg 0$ defined on a $\theta$-invariant set\footnote{A
set $A\in\mathcal{F}$ is called $\theta$-invariant if
$\theta_{t}A=A$ for all $t\in\mathbb{T}$.}
$\Omega^{\ast}\in\mathcal{F}$ with
$\mathbb{P}\left(\Omega^{\ast}\right)=1$ such that, for any random
variable $v(\omega)$ possessing the property $0\preceq
v(\omega)\preceq\alpha u(\omega)$ for all $\omega\in\Omega^{\ast}$
and deterministic $\alpha>0$, the following holds:
\begin{equation}
\label{LSD2}
\lim_{t\rightarrow\infty}\varphi\left(t,\theta_{-t}\omega,v\left(\theta_{-t}\omega\right)\right)=u(\omega),\:\omega\in\Omega^{\ast}
\end{equation}
\end{itemize}
\end{theorem}

\subsection{Properties of the Riccati RDS}
\label{prop_Ric_RDS} In this subsection we establish some
properties of the RDS $\left(\theta^{R},\varphi^{R}\right)$ modeling the
RRE.

\begin{lemma}
\label{prop_R} The RDS $\left(\theta^{R},\varphi^{R}\right)$ with state space
$\mathbb{S}_{+}^{N}$ is order-preserving. In other words,
\begin{equation}
\label{prop_R1} X\preceq
Y~\Longrightarrow~\varphi^{R}(t,\omega,X)\preceq\varphi^{R}(t,\omega,Y),\:\forall
t\in\mathbb{T}_{+},~\omega\in\Omega,~X,Y\in\mathbb{S}_{+}^{N}
\end{equation}
Also, if $Q$ is positive definite, i.e., $Q\gg 0$, it is strongly
sublinear.
\end{lemma}
\begin{proof}
We establish order-preserving.
Eqn.~(\ref{prop_R1}) holds for $t=0$, because by definition
\begin{equation}
\label{prop_R2}
\varphi^{R}(0,\omega,\cdot)=id_{S^{N}_{+}},\:\forall\omega\in\Omega
\end{equation}
Consider $t=1$. From eqn.~(\ref{RDS6}) we have
\begin{equation}
\label{prop_R3}
\varphi^{R}(1,\omega,X)=f(\omega(0),X),\:\forall\omega\in\Omega,~X\in\mathbb{S}_{+}^{N}
\end{equation}
where $f(\cdot)$ is defined in eqn.~(\ref{def_f}).
From~\cite{Bruno} (Lemma 1, part (c)), we note that, for fixed
$\gamma\in\{0,1\}$, the function
$f_{\gamma}(\cdot)=f(\gamma,\cdot)$ is order-preserving in $X$,
i.e.,
\begin{equation}
\label{prop_R4} X\preceq Y~\Longrightarrow~f_{\gamma}(X)\preceq
f_{\gamma}(Y),\:\forall X,Y\in\mathbb{S}_{+}^{N}
\end{equation}
Hence, for a given $\omega\in\Omega$, we have from
eqns.~(\ref{prop_R3},\ref{prop_R4}), if $X\preceq Y$,
\begin{eqnarray}
\label{prop_R5} \varphi^{R}(1,\omega,X)
=
f_{\omega(0)}(X)
\preceq
f_{\omega(0)}(Y)
=
\varphi^{R}(1,\omega,Y)
\end{eqnarray}
Thus, the order-preserving property is established for $t=1$. For
$t>1$, we have from eqn.~(\ref{def_RDS7})
\begin{equation}
\label{prop_R6} \varphi^{R}(t,\omega,X)=f_{\omega(t-1)}\circ
f_{\omega(t-2)}\circ\cdots\circ f_{\omega(0)}(X)
\end{equation}
For $\omega\in\Omega$, the functions
$\{f_{\omega(i)}(\cdot)\}_{0\leq t-1}$ are order-preserving by
eqn.~(\ref{prop_R4}). Since the composition of order-preserving
functions remains order-preserving, from
eqn.~(\ref{prop_R6})  the function $\varphi^{R}(t,\omega,\cdot)$
is order-preserving in $X$. This establishes the order-preserving
of the RDS $\left(\theta^{R},\varphi^{R}\right)$.

We now establish strong sublinearity  when $Q\gg 0$. Fix $\gamma\in\{0,1\}, \,\lambda\in (0,1)$. Then, from the concavity of
$f_{\gamma}(\cdot)$ (Lemma 1, part (e) in~\cite{Bruno}), we have

\begin{equation}
\label{prop_R7} \lambda
f_{\gamma}(X)+(1-\lambda)f_{\gamma}(0)\preceq f_{\gamma}(\lambda
X),\:\: \forall\,\,X\in\mathbb{S}_{+}^{N}
\end{equation}
Again, from~\cite{Bruno} (Lemma 1 part (f)), we have
\begin{equation}
\label{prop_R8} Q\preceq f_{\gamma}(0)
\end{equation}
Under the assumption $Q\gg 0$ and $\lambda\in (0,1)$, we have from
eqn.~(\ref{prop_R8})
\begin{eqnarray}
\label{prop_R9} (1-\lambda)f_{\gamma}(0)
\succeq
(1-\lambda)Q\nonumber 
\gg 0
\end{eqnarray}
From eqns.~(\ref{prop_R7},\ref{prop_R9}), we then have for every
$\lambda\in (0,1)$ and $\gamma\in\{0,1\}$
\begin{equation}
\label{prop_R10} \lambda f_{\gamma}(X)\ll f_{\gamma}(\lambda
X),\:X\in\mathbb{S}_{+}^{N}
\end{equation}
We then have from eqn.~(\ref{prop_R10}) for all
$X\in\mathbb{S}_{+}^{N}$, $\lambda\in (0,1)$, $\omega\in\Omega$, and
$t=1$
\begin{eqnarray}
\label{prop_R11} \lambda\varphi^{R}(1,\omega,X)
= \lambda
f_{\omega(0)}(X)
\ll f_{\omega(0)}(\lambda X)
= \varphi^{R}(1,\omega,\lambda X)
\end{eqnarray}
which establishes strong sublinearity for $t=1$ (the
above is \emph{stronger} than strong sublinearity, as given by
Definition~\ref{sublinearity}, since the latter requires $\ll$ to
hold only for $X\in\mathbb{S}_{++}^{N}$.) We extend it to $t>1$ by
induction. Assume that the property in eqn.~(\ref{prop_R11})
(which implies strong sublinearity) holds for $t=s>0$. We now show
that it holds for $t=s+1$. Indeed, for
$X\in\mathbb{S}_{+}^{N}$
{\small
\begin{eqnarray}
\label{prop_R12} \lambda\varphi^{R}(s+1,\omega,X)
& = &
\lambda
f_{\omega(s)}\left(\varphi^{R}(s,\omega,X)\right)
\ll
f_{\omega(s)}\left(\lambda\varphi^{R}(s,\omega,X)\right)
\\ \nonumber & \preceq &
 f_{\omega(s)}\left(\varphi^{R}(s,\omega,\lambda
X)\right)
= \varphi^{R}(s+1,\omega,\lambda X)
\end{eqnarray}
}
where the second step follows from eqn.~(\ref{prop_R10}) and the
third  from the induction step
\begin{equation}
\label{prop_R13}
\lambda\varphi^{R}(s,\omega,X)\ll\varphi^{R}(s,\omega,\lambda X)
\end{equation}
and the fact that $f_{\omega(s)}(\cdot)$ is order-preserving. Thus
we have strong sublinearity.
\end{proof}


\section{Proofs of Theorems~\ref{main_th},\ref{supp_inv}}
\label{main_proof}

\subsection{Proof of Theorem~\ref{main_th}}
\label{proof_main_th}



\begin{lemma}
\label{step} Consider the RDS $\left(\theta^{R},\varphi^{R}\right)$. Assume:
$\overline{\gamma}\in
(\overline{\gamma}^{\mbox{\scriptsize sb}},1]$; $Q$ positive
definite. Then there exists  unique almost equilibrium
$u^{\overline{\gamma}}(\omega)\gg 0$ defined on a
$\theta^{R}$-invariant set $\Omega^{\ast}\in\mathcal{F}$ with
$\mathbb{P}^{\overline{\gamma}}\left(\Omega^{\ast}\right)=1$ s.t.~for any random variable $v(\omega)$ possessing the property
$0\preceq v(\omega)\preceq\alpha u^{\overline{\gamma}}(\omega) \,\forall\, \omega\in\Omega^{\ast}$ and deterministic $\alpha>0$, the
following holds:
\begin{equation}
\label{step1}
\lim_{t\rightarrow\infty}\varphi^{R}
\left(t,\theta^{R}_{-t}\omega,v\left(\theta^{R}_{-t}\omega\right)
\right)=u^{\overline{\gamma}}(\omega),\:\omega\in\Omega^{\ast}
\end{equation}
\end{lemma}

\begin{proof}
From Lemma~\ref{prop_R}, $\left(\theta^{R},\varphi^{R}\right)$ is
strongly sublinear  and order-preserving. It is conditionally
compact because the space $\mathbb{S}_{+}^{N}$ is finite
dimensional. Also, the cone $\mathbb{S}_{+}^{N}$ satisfies the
conditions required in the hypothesis of Theorem~\ref{LSD}. From
Lemma~1f) in~\cite{Bruno}, we note for $t>0$
\begin{eqnarray}
\label{step2} \varphi^{R}\left(t,\omega,0\right)
=
f_{\omega(t-1)}\left(\varphi^{R}(t-1,\omega,0)\right)
\succeq  Q
\gg
 0
\end{eqnarray}
Thus the hypotheses of Theorem~\ref{LSD} are satisfied and
precisely one of the assertions~\textbf{a)} or~\textbf{b)} holds.
We show assertion~\textbf{a)} does not hold. Assume
that~\textbf{a)} holds on the contrary. Then, there exists
$P_{0}\in\mathbb{S}_{+}^{N}$ such that
\begin{equation}
\label{step3}
\mathbb{P}^{\overline{\gamma}}\left(\lim_{t\rightarrow\infty}
\left\|\varphi^{R}\left(t,\theta^{R}_{-t}\omega,P_{0}\right)\right\|=
\infty\right)=1
\end{equation}
Then, for every $N\in\mathbb{T}_{+}$, we have
\begin{equation}
\label{step4}
\lim_{t\rightarrow\infty}\left\|\varphi^{R}
\left(t,\theta^{R}_{-t}\omega,P_{0}\right)\right\|>N,
\:\:\mathbb{P}^{\overline{\gamma}}\:\mbox{a.s.}
\end{equation}
In other words, the sequence
$\left\{\mathbb{I}_{(N,\infty)}\left(\left\|\varphi^{R}
\left(t,\theta^{R}_{-t}\omega,P_{0}\right)
\right\|\right)\right\}_{t\in\mathbb{T}_{+}}$
satisfies
\begin{equation}
\label{step5}
\mathbb{P}^{\overline{\gamma}}\left(\lim_{t\rightarrow\infty}
\mathbb{I}_{(N,\infty)}\left(\left\|\varphi^{R}
\left(t,\theta^{R}_{-t}\omega,P_{0}\right)\right\|\right)=1\right)=1
\end{equation}
Since convergence a.s.~implies convergence in probability, we have
for every $\varepsilon>0$,
\begin{equation}
\label{step6}
\lim_{t\rightarrow\infty}\mathbb{P}^{\overline{\gamma}}
\left(\left|\mathbb{I}_{(N,\infty)}\left(\left\|\varphi^{R}
\left(t,\theta^{R}_{-t}\omega,P_{0}\right)\right\|\right)-
1\right|\leq\varepsilon\right)=1
\end{equation}
Since $\mathbb{I}_{(N,\infty)}(\cdot)$ takes the values
$\{0,1\}$, eqn.~(\ref{step6}) implies
\begin{equation}
\label{step7}
\lim_{t\rightarrow\infty}\mathbb{P}^{\overline{\gamma}}
\left(\mathbb{I}_{(N,\infty)}\left(\left\|\varphi^{R}
\left(t,\theta^{R}_{-t}\omega,P_{0}\right)\right\|\right)=1\right)=1
\end{equation}
We thus have
\begin{equation}
\label{step8}
\lim_{t\rightarrow\infty}\mathbb{P}^{\overline{\gamma}}
\left(\left\|\varphi^{R}\left(t,\theta^{R}_{-t}\omega,P_{0}\right)
\right\|>N\right)=1
\end{equation}
Since the above holds for every $N\in\mathbb{T}_{+}$, we have
\begin{equation}
\label{step9}\lim_{N\rightarrow\infty}\sup_{t\in\mathbb{T}_{+}}
\mathbb{P}^{\overline{\gamma}}\left(\left\|\varphi^{R}
\left(t,\theta^{R}_{-t}\omega,P_{0}\right)\right\|>N\right)=1
\end{equation}
On the other hand,
$\overline{\gamma}>\overline{\gamma}^{\mbox{\scriptsize sb}}$ and
Lemma~\ref{orbit_lemma} both imply
\begin{eqnarray}
\label{step10}
\lim_{N\rightarrow\infty}\sup_{t\in\mathbb{T}_{+}}\mathbb{P}^{\overline{\gamma}}
\left(\left\|\varphi^{R}\left(t,\theta^{R}_{-t}\omega,P_{0}\right)\right\|>N\right)
& = &
\lim_{N\rightarrow\infty}\sup_{t\in\mathbb{T}_{+}}\mathbb{P}^{\overline{\gamma}}
\left(\left\|\varphi^{R}\left(t,\omega,P_{0}\right)\right\|>N\right)\nonumber
\\ & = & \lim_{N\rightarrow\infty}\sup_{t\in\mathbb{T}_{+}}
\mathbb{P}^{\overline{\gamma},P_{0}}\left(\left\|P_{t}\right\|>N\right)\nonumber
\\ & = & 0
\end{eqnarray}
This contradicts eqn.~(\ref{step4}) and~\textbf{a)} does not hold. Thus~\textbf{b)}
holds, and we have the result.
\end{proof}

Lemma~\ref{step} establishes the existence of a unique almost
equilibrium $u^{\overline{\gamma}}$ if
$\overline{\gamma}^{\mbox{\scriptsize sb}}<\overline{\gamma}\leq 1$. From the
distributional equivalence of pull-back and forward orbits, it
follows that, for the Markov-Feller pair
$(L^{\overline{\gamma}},T^{\overline{\gamma}})$,
$T^{\overline{\gamma}}$ is uniquely ergodic. However, to show that
the measure induced by $u^{\overline{\gamma}}$ on
$\mathbb{S}^{N}_{+}$ is attracting for $T^{\overline{\gamma}}$,
eqn.~(\ref{step1}) must hold for all initial $v$.
Lemma~\ref{step} establishes convergence for a restricted class of
initial conditions $v$. We need 
  to extend it to general initial conditions.
\begin{lemma}
\label{eq} For $\overline{\gamma}^{\mbox{\scriptsize sb}}<\overline{\gamma}\leq 1$
let $u^{\overline{\gamma}}$ be an almost equilibrium of the RDS
$\left(\theta^{R},\varphi^{R}\right)$. Then
\begin{equation}
\label{eq1}
\mathbb{P}^{\overline{\gamma}}\left(\omega:u^{\overline{\gamma}}(\omega)\succeq
Q\right)=1
\end{equation}
\end{lemma}
\begin{proof}
By the definition of an almost equilibrium (see
Definition~\ref{equilibrium}) we have
\begin{eqnarray}
\label{eq2}
\mathbb{P}^{\overline{\gamma}}\left(\omega:\varphi^{R}
\left(1,\omega,u^{\overline{\gamma}}(\omega)\right)\succeq
Q\right)
=
\mathbb{P}^{\overline{\gamma}}
\left(\omega:u^{\overline{\gamma}}(\theta_{1}^{R\overline{\gamma}}\omega)\succeq
Q\right)
=
\mathbb{P}^{\overline{\gamma}}\left(\omega:u^{\overline{\gamma}}(\omega)\succeq
Q\right)
\end{eqnarray}
Again, by Lemma~1f) in~\cite{Bruno}, we have
$\mathbb{P}^{\overline{\gamma}}$~a.s.
\begin{eqnarray}
\label{eq3}
\varphi^{R}\left(1,\omega,u^{\overline{\gamma}}(\omega)\right)
 =
  f_{\omega(0)}(u^{\overline{\gamma}}(\omega))
\succeq    Q
\end{eqnarray}
Since eqn.~(\ref{eq3}) holds
$\mathbb{P}^{\overline{\gamma}}$~a.s., we have
\begin{equation}
\label{eq4}
\mathbb{P}^{\overline{\gamma}}\left(\omega:\varphi^{R}
\left(1,\omega,u^{\overline{\gamma}}(\omega)\right)\succeq
Q\right)=1
\end{equation}
The Lemma then follows from eqns.~(\ref{eq2},\ref{eq4}).
\end{proof}
\begin{proof}[Proof of Theorem~\ref{main_th}] We now complete the proof of Theorem~\ref{main_th}. The key step consists of finding a suitable modification $\widetilde{X}(\omega)$ of the initial condition $P_{0}$, such that $\widetilde{X}(\omega)=P_{0}$ a.s. and there exists a deterministic $\alpha>0$ satisfying $0\preceq\widetilde{X}(\omega)\preceq\alpha u^{\overline{\gamma}}(\omega)$. In that case, we can invoke Lemma~\ref{step} to establish weak convergence of the sequence $\{\varphi^{R}\left(t,\omega,\widetilde{X}(\omega)\right)\}_{t\in\mathbb{T}_{+}}$ with initial condition $\widetilde{X}(\omega)$ to $\mathbb{\mu}^{\overline{\gamma}}$. Since $\widetilde{X}(\omega)$ is a.s. equal to $P_{0}$, this would allow us to deduce the weak convergence of the desired sequence $\{\varphi^{R}\left(t,\omega,P_{0}\right)\}_{t\in\mathbb{T}_{+}}$. We detail such a construction in the following.

For $\overline{\gamma}>\overline{\gamma}^{\mbox{\scriptsize{sb}}}$ (note that $\overline{\gamma}^{\mbox{\scriptsize{sb}}}=0$ under the assumptions of Theorem~\ref{main_res}) let $\mu^{\overline{\gamma}}$
be the distribution of the unique almost equilibrium in
Lemma~\ref{step}. By Lemma~\ref{eq} we have
\begin{equation}
\label{main_res3}
\mu^{\overline{\gamma}}\left(\mathbb{S}_{++}^{N}\right)=1
\end{equation}
since $u^{\overline{\gamma}}(\omega)\preceq Q\gg 0$ a.s.
Let $P_{0}\in\mathbb{S}_{+}^{N}$ be an arbitrary initial state. By
construction of the RDS $\left(\theta^{R},\varphi^{R}\right)$, the sequences
$\left\{P_{t}\right\}_{t\in\mathbb{T}_{+}}$ and
$\{\varphi^{R}\left(t,\omega,P_{0}\right)\}_{t\in\mathbb{T}_{+}}$
are distributionally equivalent, i.e.,
\begin{equation}
\label{main_res4}
P_{t}\ndtstile{}{d}\varphi^{R}\left(t,\omega,P_{0}\right)
\end{equation}
Recall $\Omega^{\ast}$ as the $\theta^{R}$-invariant set with
$\mathbb{P}^{\overline{\gamma}}(\Omega^{\ast})=1$ in
Lemma~\ref{step} on which the almost equilibrium
$u^{\overline{\gamma}}$ is defined. By Lemma~\ref{eq}, there
exists $\Omega_{1}\subset\Omega^{\ast}$ with
$\mathbb{P}^{\overline{\gamma}}(\Omega_{1})=1$, such that
\begin{equation}
\label{main_res5} u^{\overline{\gamma}}(\omega)\succeq
Q,\:\omega\in\Omega_{1}
\end{equation}
Define the random variable
$\widetilde{X}:\Omega\longmapsto\mathbb{S}_{+}^{N}$ by
\begin{equation}
\label{main_res6} \left\{ \begin{array}{ll}
                    P_{0} & \mbox{if $\omega\in\Omega_{1}$} \\
                    0 & \mbox{if $\omega\in\Omega_{1}^{c}$}
                   \end{array}
          \right.
\end{equation}
Now choose $\alpha>0$ sufficiently large, such that,
\begin{equation}
\label{main_res7} P_{0}\preceq\alpha Q
\end{equation}
This is possible because $Q\gg 0$. Then
\begin{equation}
\label{main_res8} 0\preceq\widetilde{X}(\omega)\preceq\alpha
u^{\overline{\gamma}}(\omega),\:\omega\in\Omega^{\ast}
\end{equation}
Indeed, we have
\begin{eqnarray}
\label{main_res9} 0\preceq
P_{0}&=&\widetilde{X}(\omega)\preceq\alpha Q\preceq\alpha
u^{\overline{\gamma}}(\omega),\:\omega\in\Omega_{1}\\
\label{main_res10} 0=\widetilde{X}(\omega)&\preceq& \alpha
u^{\overline{\gamma}}(\omega),\:\omega\in\Omega\backslash\Omega_{1}
\end{eqnarray}
Then, by Lemma~\ref{step}
\begin{equation}
\label{main_res11}
\lim_{t\rightarrow\infty}\varphi^{R}\left(t,\theta^{R}_{-t}\omega,\widetilde{X}
\left(\theta^{R}_{-t}\omega\right)\right)=
u^{\overline{\gamma}}(\omega),\:\omega\in\Omega^{\ast}
\end{equation}
Since convergence $\mathbb{P}^{\overline{\gamma}}$~a.s.~implies
convergence in distribution, we have
\begin{equation}
\label{main_res12}
\varphi^{R}\left(t,\theta^{R}_{-t}\omega,\widetilde{X}
\left(\theta^{R}_{-t}\omega\right)\right)
\Longrightarrow\mu^{\overline{\gamma}}
\end{equation}
as $t\rightarrow\infty$, where $\Longrightarrow$ denotes weak
convergence or convergence in distribution. Then by
Lemma~\ref{orbit_lemma}, the sequence
$\left\{\varphi^{R}\left(t,\omega,\widetilde{X}(\omega)\right)
\right\}_{t\in\mathbb{T}_{+}}$
also converges in distribution to the unique stationary
distribution $\mu^{\overline{\gamma}}$, i.e., as
$t\rightarrow\infty$
\begin{equation}
\label{main_res13}
\varphi^{R}\left(t,\omega,\widetilde{X}(\omega)
\right)\Longrightarrow\mu^{\overline{\gamma}}
\end{equation}

Now, since $\mathbb{P}^{\overline{\gamma}}(\Omega_{1})=1$, by
eqn.~(\ref{main_res6})
\begin{equation}
\label{main_res14}
\varphi^{R}\left(t,\omega,P_{0}\right)=
\varphi^{R}\left(t,\omega,\widetilde{X}(\omega)\right),
\:\mathbb{P}^{\overline{\gamma}}\:a.s.,\:t\in\mathbb{T}_{+}
\end{equation}
which implies
\begin{equation}
\label{main_res15}
\varphi^{R}\left(t,\omega,P_{0}\right)
\ndtstile{}{d}\varphi^{R}\left(t,\omega,\widetilde{X}(\omega)\right),
\:t\in\mathbb{T}_{+}
\end{equation}
From eqns.~(\ref{main_res13},\ref{main_res15}), we then have as
$t\rightarrow\infty$
\begin{equation}
\label{main_res16a}
\varphi^{R}\left(t,\omega,P_{0}\right)\Longrightarrow\mu^{\overline{\gamma}}
\end{equation}
which together with eqn.~(\ref{main_res4}) implies
\begin{equation}
\label{main_res16} P_{t}\Longrightarrow\mu^{\overline{\gamma}}
\end{equation}
as $t\rightarrow\infty$. This completes part~i) of
Theorem~\ref{main_th}.

For part~ii), we note that, if
$\overline{\gamma}>\overline{\gamma}^{\mbox{\scriptsize bim}}$, by~\cite{Bruno} (see
Result~\ref{bruno_res}), there exists
$M^{\overline{\gamma},P_{0}}$, such that
\begin{equation}
\label{main_res17}
\sup_{t\in\mathbb{T}_{+}}\mathbb{E}^{\overline{\gamma},P_{0}}[P_{t}]\preceq
M^{\overline{\gamma},P_{0}}
\end{equation}
From eqn.~(\ref{main_res11}), we have by Fatou's lemma
\begin{equation}
\label{main_res18}
\liminf_{t\rightarrow\infty}\mathbb{E}^{\overline{\gamma}}
\left[\left\|\varphi^{R}\left(t,\theta^{R}_{-t}\omega,\widetilde{X}
\left(\theta^{R}_{-t}\omega\right)\right)\right\|\right]
\geq
\mathbb{E}^{\overline{\gamma}}\left[\left\|u(\omega)\right\|\right]
\end{equation}
Using Lemma~\ref{orbit_lemma} and standard manipulations
{\small
\begin{eqnarray}
\label{main_res19}
\int_{S^{n}_{+}}\left\|Y\right\|\mu^{\overline{\gamma}}(dY) & = &
\mathbb{E}^{\overline{\gamma}}\left[\left\|u(\omega)\right\|\right]
\nonumber
\leq
\liminf_{t\rightarrow\infty}\mathbb{E}^{\overline{\gamma}}
\left[\left\|\varphi^{R}\left(t,\theta^{R}_{-t}\omega,\widetilde{X}
\left(\theta^{R}_{-t}\omega\right)\right)\right\|\right]\nonumber
\\ & = & \liminf_{t\rightarrow\infty}\mathbb{E}^{\overline{\gamma}}
\left[\left\|\varphi^{R}\left(t,\omega,\widetilde{X}(\omega)
\right)\right\|\right]\nonumber
 =  \liminf_{t\rightarrow\infty}\mathbb{E}^{\overline{\gamma}}
\left[\left\|\varphi^{R}\left(t,\omega,P_{0}\right)\right\|\right]\nonumber
\\ & = &
\liminf_{t\rightarrow\infty}\mathbb{E}^{\overline{\gamma}}
\left[\left\|P_{t}\right\|\right]\nonumber
\leq \liminf_{t\rightarrow\infty}\mathbb{E}^{\overline{\gamma}}
\left[\mbox{Tr}\,\left(P_{t}\right)\right]\leq
\mbox{Tr}\,M^{\overline{\gamma},P_{0}}
<
\infty
\end{eqnarray}
} which establishes part ii).
\end{proof}

\subsection{Proof of Theorem~\ref{supp_inv}}
\label{proof_supp_inv} We state a result on the properties of
invariant probabilities of Markov-Feller operators needed for the
proof. Consider a locally compact separable metric space
$(\mathcal{X},d)$ and define the topological lower limit of a
sequence $\left\{B_{n}\right\}_{n\in\mathbb{T}_{+}}$ of subsets of
$\mathcal{X}$ by
\begin{equation}
\label{supp_inv20}
\mbox{Li}_{n\rightarrow\infty}B_{n}=\left\{x\in\mathcal{X}
\left|\right.\exists\:\mbox{a sequence}\:\left\{x_{n}\right\},\:\mbox{s.t.}\: x_{n}\in B_{n},\:\forall n,\:\mbox{and $\left\{x_{n}\right\}$ converges to $x$}\right\}
\end{equation}
which, by definition, is closed. We then have the following result
from~\cite{Zaharopol}.
\begin{theorem}[Theorem 1.3.1~\cite{Zaharopol}]
\label{th_Zaharopol} Let $(L,T)$ be a Markov-Feller pair defined
on $(\mathcal{X},d)$. For all $x\in\mathcal{X}$, consider the
sequence of measures $\left\{T^{t}\delta_{x}\right\}_{t\in\mathbb{T}_{+}}$
and define
\begin{eqnarray}
\label{th_Zaharopol1}
\sigma(x)&=&\mbox{Li}_{t\rightarrow\infty}\mbox{supp}\left(T^{t}\delta_{x}\right)\\
\label{th_Zaharopol2} \sigma &=& \cap_{x\in\mathcal{X}}\sigma(x)
\end{eqnarray}
Then, if $(L,T)$ has an attractive probability $\mu$, we have
\begin{equation}
\label{th_Zaharopol3} \mbox{supp}(\mu)=\sigma
\end{equation}
\end{theorem}

We now complete the proof of Theorem~\ref{supp_inv}.

\begin{proof}[Proof of Theorem~\ref{supp_inv}]
Fix $\overline{\gamma}^{\mbox{\scriptsize sb}}<\overline{\gamma}<1$ and recall the Markov-Feller pair $\left(L^{\overline{\gamma}},T^{\overline{\gamma}}\right)$
(eqns.~(\ref{def_Lgamma},\ref{def_Tgamma}) and
Proposition~\ref{prop_MarkFell}.) By Theorem~\ref{main_th},
$\mu^{\overline{\gamma}}$ is an attractive probability for the
pair $\left(L^{\overline{\gamma}},T^{\overline{\gamma}}\right)$. We now use
Theorem~\ref{th_Zaharopol} to obtain the support of
$\mu^{\overline{\gamma}}$.

For $X\in\mathbb{S}^{N}_{+}$, let
\begin{equation}
\label{supp_inv21}
\sigma(X)=\mbox{Li}_{t\rightarrow\infty}\mbox{supp}
\left(\left(T^{\overline{\gamma}}\right)^{t}\delta_{X}\right)
\end{equation}
Then
\begin{equation}
\label{supp_inv23}
\mbox{supp}\left(\mu^{\overline{\gamma}}\right)=
\cap_{X\in\mathbb{S}^{N}_{+}}\sigma(X)
\end{equation}
We first show that
\begin{equation}
\label{supp_inv24}
\mbox{cl}(\mathcal{S})\subset\mbox{supp}\left(\mu^{\overline{\gamma}}\right)
\end{equation}
where
\begin{equation}
\label{supp_inv25} \mathcal{S}=\left\{f_{i_{1}}\circ
f_{i_{2}}\circ\cdots\circ
f_{i_{s}}\left(P^{\ast}\right)\left|\right.i_{r}\in\{0,1\},\:1\leq r\leq
s,\:s\in\mathbb{T}_{+}\right\}
\end{equation}

To this end, consider $X\in\mathbb{S}_{+}^{N}$. It follows from
the properties of $T^{\overline{\gamma}}$ that
\begin{equation}
\label{supp_inv22}
\mbox{supp}\left(\left(T^{\overline{\gamma}}\right)^{t}\delta_{X}\right)=
\left\{f_{i_{1}}\circ f_{i_{2}}\circ\cdots\circ
f_{i_{t}}\left(X\right)\left|\right.i_{r}\in\{0,1\},\:1\leq r\leq
t\right\}
\end{equation}
Indeed, starting from $X$, the only set of points 
reached with non-zero probability (and exhaustively) are the ones
obtained by applying $t$ arbitrary compositions of $f_{0}$ and
$f_{1}$ on $X$. Since a set with finite cardinality is closed, we
have the R.H.S. of eqn.~(\ref{supp_inv22}) as the support of
$\left(T^{\overline{\gamma}}\right)^{t}$.

Recall $P^{\ast}\in\mathbb{S}^{N}_{++}$ to be the deterministic
fixed point of the algebraic Riccati equation. Consider the
sequence $\left\{f_{1}^{t}(X)\right\}_{t\in\mathbb{T}_{+}}$ in
$\mathbb{S}^{N}_{+}$. It follows from eqn.~(\ref{supp_inv22}) that
\begin{equation}
\label{supp_inv26} f_{1}^{t}(X)\in
\mbox{supp}\left(\left(T^{\overline{\gamma}}\right)^{t}\delta_{X}\right),\:\forall
t
\end{equation}
Also, from the properties of the algebraic Riccati equation, we
have
\begin{equation}
\label{supp_inv27} \lim_{t\rightarrow\infty}f_{1}^{t}(X)=P^{\ast}
\end{equation}
Hence, by the definition of topological lower limit,
\begin{equation}
\label{supp_inv28} P^{\ast}\in\sigma(X)
\end{equation}
We now show that, for every $s\in\mathbb{T}_{+}$ and
$i_{r}\in\{0,1\},\:1\leq r\leq s$,
\begin{equation}
\label{supp_inv29} f_{i_{1}}\circ f_{i_{2}}\circ\cdots\circ
f_{i_{s}}\left(P^{\ast}\right)\in\sigma(X)
\end{equation}
Indeed, define the sequence, $\left\{X_{t}\right\}_{t\in\mathbb{T}_{+}}$ as
\begin{equation}
\label{supp_inv30} X_{t}=\left\{ \begin{array}{ll}
                    f_{i_{1}}\circ f_{i_{2}}\circ\cdots\circ
f_{i_{t}}\left(X\right) & \mbox{if $t\leq s$} \\
                    f_{i_{1}}\circ f_{i_{2}}\circ\cdots\circ
f_{i_{s}}\circ f_{1}^{t-s}(X) & \mbox{if $t>s$}
                   \end{array}
          \right.
\end{equation}
Clearly,
$X_{t}\in\mbox{supp}\left(\left(T^{\overline{\gamma}}\right)^{t}\delta_{X}\right),
\:\forall t$. Since the sequence $\left\{f_{1}^{t-s}(X)\right\}_{t>s}$ converges to $P^{\ast}$, we have
\begin{eqnarray}
\label{supp_inv31} \lim_{t\rightarrow\infty}X_{t} & = &
\lim_{t>s,t\rightarrow\infty}X_{t}\nonumber 
=
\lim_{t>s,t\rightarrow\infty}f_{i_{1}}\circ
f_{i_{2}}\circ\cdots\circ
f_{i_{s}}\left(f_{1}^{t-s}(X)\right)\nonumber \\ & = &
f_{i_{1}}\circ f_{i_{2}}\circ\cdots\circ
f_{i_{s}}\left(\lim_{t>s,t\rightarrow\infty}f_{1}^{t-s}(X)\right)\nonumber
=  f_{i_{1}}\circ f_{i_{2}}\circ\cdots\circ
f_{i_{s}}\left(P^{\ast}\right)
\end{eqnarray}
where the continuity of $f_{i_{1}}\circ f_{i_{2}}\circ\cdots\circ
f_{i_{s}}$ (being the composition of continuous functions,
Proposition~\ref{prop1}) permits bringing the limit inside.

Thus, the sequence $\left\{X_{t}\right\}_{t\in\mathbb{T}_{+}}$
converges to $f_{i_{1}}\circ f_{i_{2}}\circ\cdots\circ
f_{i_{s}}\left(P^{\ast}\right)$ and
$X_{t}\in\mbox{supp}\left(\left(T^{\overline{\gamma}}\right)^{t}\delta_{X}\right),
\:\forall
t$. Hence, $f_{i_{1}}\circ f_{i_{2}}\circ\cdots\circ
f_{i_{s}}\left(P^{\ast}\right)\in\sigma(X)$. It then follows
\begin{equation}
\label{supp_inv32} \mathcal{S}\subset\sigma(X),\:\forall
X\in\mathbb{S}_{+}^{N}
\end{equation}
Since the set $\sigma(X)$ is closed, we have
\begin{equation}
\label{supp_inv33}
\mbox{cl}\left(\mathcal{S}\right)\subset\sigma(X),\:\forall
X\in\mathbb{S}_{+}^{N}
\end{equation}
which implies by eqn.~(\ref{supp_inv23})
\begin{equation}
\label{supp_inv34}
\mbox{cl}(\mathcal{S})\subset\mbox{supp}\left(\mu^{\overline{\gamma}}\right)
\end{equation}
To obtain the reverse inclusion, we note that
\begin{eqnarray}
\label{supp_inv35} \sigma\left(P^{\ast}\right) 
=
\mbox{Li}_{t\rightarrow\infty}\mbox{supp}
\left(\left(T^{\overline{\gamma}}\right)^{t}
\delta_{P^{\ast}}\right)\nonumber
\subset
\mbox{cl}\left(\cup_{t\in\mathbb{T}_{+}}
\mbox{supp}\left(\left(T^{\overline{\gamma}}\right)^{t}
\delta_{P^{\ast}}\right)\right)\nonumber
= \mbox{cl}(\mathcal{S})
\end{eqnarray}
Here the first step follows from the fact that, if
$Y\in\mbox{Li}_{t\rightarrow\infty}\mbox{supp}
\left(\left(T^{\overline{\gamma}}\right)^{t}
\delta_{P^{\ast}}\right)$,
then $Y$ is a limit point of
$\cup_{t\in\mathbb{T}_{+}}\mbox{supp}\left(\left(T^{\overline{\gamma}}\right)^{t}
\delta_{P^{\ast}}\right)$
and hence belongs to its closure. The last step is obvious from
eqn.~(\ref{supp_inv22}).
We thus have
\begin{eqnarray}
\label{supp_inv36} \mbox{supp}\left(\mu^{\overline{\gamma}}\right)
 \subset
 \sigma\left(P^{\ast}\right)
 \subset
  \mbox{cl}(\mathcal{S})
\end{eqnarray}
which establishes the other inclusion and we have
\begin{equation}
\label{supp_inv37}
\mbox{supp}\left(\mu^{\overline{\gamma}}\right)=\mbox{cl}(\mathcal{S})
\end{equation}

It remains to establish eqn.~(\ref{main_th1}). To this end, we
first show that
\begin{equation}
\label{supp_inv38} f_{i_{1}}\circ f_{i_{2}}\circ\cdots\circ
f_{i_{s}}\left(P^{\ast}\right)\succeq P^{\ast}\:\: \forall\,\,s\in\mathbb{T}_{+},\,\, i_{r}\in\{0,1\},\,\,1\leq r\leq s
\end{equation}
We prove this by an inductive argument on $s$. Clearly, the above
holds for $s=1$, as
\begin{equation}
\label{supp_inv39} f_{0}\left(P^{\ast}\right)=AP^{\ast}A^{T}+Q\succeq
P^{\ast}
\end{equation}
($A$ is unstable and $Q>> 0$) and
\begin{equation}
\label{supp_inv40} f_{1}\left(P^{\ast}\right)=P^{\ast}
\end{equation}
Assume the claim holds for $s\leq t$. We now show  it
holds for $s=t+1$. By the induction step,
\begin{equation}
\label{supp_inv41} f_{i_{2}}\circ\cdots\circ
f_{i_{t+1}}\left(P^{\ast}\right)\succeq P^{\ast}
\end{equation}
If $i_{1}=0$
\begin{eqnarray}
\label{supp_inv42a} f_{i_{1}}\circ f_{i_{2}}\circ\cdots\circ
f_{i_{t+1}}\left(P^{\ast}\right)
=
 f_{0}\left(f_{i_{2}}\circ\cdots\circ
f_{i_{t+1}}\left(P^{\ast}\right)\right)
\succeq
f_{0}\left(P^{\ast}\right)
\succeq
P^{\ast}
\end{eqnarray}
which follows from the order preserving property of $f_{1}$ and
eqn.~(\ref{supp_inv39}).
If $i_{1}=1$
\begin{eqnarray}
\label{supp_inv42} f_{i_{1}}\circ f_{i_{2}}\circ\cdots\circ
f_{i_{t+1}}\left(P^{\ast}\right)
= f_{1}\left(f_{i_{2}}\circ\cdots\circ
f_{i_{t+1}}\left(P^{\ast}\right)\right)
\succeq
f_{1}\left(P^{\ast}\right)\nonumber
=
 P^{\ast}
\end{eqnarray}
which follows from the order preserving property of $f_{0}$ and
eqn.~(\ref{supp_inv40}).

We thus have for $i_{r}\in\{0,1\},\:1\leq r\leq t+1$
\begin{equation}
\label{supp_inv43} f_{i_{1}}\circ f_{i_{2}}\circ\cdots\circ
f_{i_{t+1}}\left(P^{\ast}\right)\succeq P^{\ast}
\end{equation}
and the claim in eqn.~(\ref{supp_inv38}) follows.

To complete the proof, we note from the above,
\begin{equation}
\label{supp_inv44} \mathcal{S}\subset \left[P^{\ast},\infty\right)
\end{equation}
Since the conic interval $\left[P^{\ast},\infty\right)$ is closed, we have
\begin{equation}
\label{supp_inv45} \mbox{cl}(\mathcal{S})\subset \left[P^{\ast},\infty\right)
\end{equation}
and eqn.~(\ref{main_th1}) follows as
\begin{equation}
\label{supp_inv46}
\mu^{\overline{\gamma}}\left(\mbox{cl}(\mathcal{S})\right)=1
\end{equation}
$\mbox{cl}(\mathcal{S})$ being the support of
$\mu^{\overline{\gamma}}$.

\end{proof}

\section{A scalar example and numerical studies}
\subsection{Scalar Example} \label{scal_num} We investigate in detail a scalar
system, for which we qualitatively characterize the structure of
the support of the invariant distributions. In the general
(non-scalar) case, Theorem~\ref{supp_inv} explicitly characterizes
the support set of the invariant distributions and shows, in
particular, that $\mbox{supp}\left(\mu^{\overline{\gamma}}\right)$ is
independent of $\overline{\gamma}$ as long as
$\overline{\gamma}^{\mbox{\scriptsize sb}}<\overline{\gamma}<1$. In this section, by
studying a scalar example, we show that the set $\mathcal{S}$ is,
in general, not a dense subset of the conic interval
$\left[P^{\ast},\infty\right)$. In fact, for the example we consider, the
support set $\mathcal{S}$ is a highly fractured subset of the
interval $\left[P^{\ast},\infty\right)$ with a self-similar structure (to be
explained below), thus exhibiting fractal-like properties.

Consider the scalar system, $A=\sqrt{2}$, $C=Q=R=1$. The
functions $f_{0},f_{1}$ then reduce to
\begin{equation}
\label{scal_num1} f_{0}(X)=2X+1
\end{equation}
\begin{equation}
\label{scal_num2} f_{1}(X)=3-\frac{2}{X+1}
\end{equation}
and the fixed point of the algebraic Riccati equation is given by
\begin{equation}
\label{scal_num3} P^{\ast}=1+\sqrt{2}
\end{equation}
The next proposition shows that $\mbox{supp}\left(\mu^{\overline{\gamma}}\right)$ is not dense in
$[1+\sqrt{2},\infty)$ and exhibits self-similarity.
\begin{proposition}
\label{prop_scal} Define 
$\mathcal{S}_{0}=\mbox{cl}\,(\mathcal{S})\cap [1+\sqrt{2},3]$
 and, recursively, 
   $\mathcal{S}_{n}=\left\{2Y+1,\:Y\in\mathcal{S}_{n-1}\right\},\:n\geq 1$.
 We then have
$\mbox{cl}\,\mathcal{S}=\cup_{n\geq 0}\mathcal{S}_{n}$
\end{proposition}
Before proving Proposition~\ref{prop_scal}, we
interpret it. First, it reflects the
self-similarity of
$\mbox{supp}\left(\mu^{\overline{\gamma}}\right)$, i.e., it
suffices to know the structure of $\mbox{supp}\left(\mu^{\overline{\gamma}}\right)$ in the interval
$[1+\sqrt{2},3]$; this structure (with proper scaling) is
repeated over space. In particular, if $\mathcal{S}_{0}$ is the
restriction of $\mbox{supp}\left(\mu^{\overline{\gamma}}\right)$ to
$[1+\sqrt{2},3]$, the restriction to
$[3+2\sqrt{2},7]$ is given by $\mathcal{S}_{1}$ (which can be
written alternatively as $f_{0}\left(\mathcal{S}_{0}\right)$) and
is a stretched version of $\mathcal{S}_{0}$, the stretching factor
being 2. More generally, for $n\geq 1$, the restriction of
$\mbox{supp}\left(\mu^{\overline{\gamma}}\right)$ to
$\left[2^{n}(1+\sqrt{2})+2^{n}-1,2^{n}.3+2^{n}-1\right]$ is given by
$\mathcal{S}^{n}$, which is a stretched version of
$\mathcal{S}_{0}$, the stretching factor being $2^{n}$. Thus,  $\mbox{supp}\left(\mu^{\overline{\gamma}}\right)$ consists of stretching
the set $\mathcal{S}_{0}$ by factors of $2^{n}$ and placing them
over the real line.

Proposition~\ref{prop_scal} also shows that
$\mbox{supp}\left(\mu^{\overline{\gamma}}\right)$ is not dense in
$\left[1+\sqrt{2},\infty\right)$ and contains holes. In fact,  for every $n\geq 0$, there is a hole of length
$2^{n+1}\sqrt{2}$ between the sets $\mathcal{S}_{n}$ and
$\mathcal{S}_{n+1}$, corresponding to the interval
$\left(2^{n}.3+2^{n}-1,2^{n+1}(1+\sqrt{2})+2^{n+1}-1\right)$. In other words,
for every $n\geq 0$, the open interval
$\left(2^{n}.3+2^{n}-1,2^{n+1}(1+\sqrt{2})+2^{n+1}-1\right)$ does not belong
to $\mbox{supp}\left(\mu^{\overline{\gamma}}\right)$.

\begin{proof}[Proof of Proposition~\ref{prop_scal}.]
Let $Y\in\mbox{supp}\left(\mu^{\overline{\gamma}}\right)$ and
$Y>3$.\footnote{If $Y\leq 3$, then $Y\in\mathcal{S}_{0}$ trivially
by construction and hence $Y\in\cup_{n\geq 0}\mathcal{S}_{n}$.}
Then, there exists a sequence $\left\{Y_{n}\right\}_{n\in\mathbb{T}_{+}}$
converging to $Y$, s.t.~$Y_{n}\in\mathcal{S}$ for all $n$.
For every $n\in\mathbb{T}_{+}$, $Y_{n}$ can be
represented as
\begin{equation}
\label{prop_scal4} Y_{n}=f_{i_{1}}\circ\cdots\circ
f_{i_{s_{n}}}\left(P^{\ast}\right)
\end{equation}
for some $s_{n}\in\mathbb{T}_{+}$ (depending on $n$) and
$i_{r}\in\{0,1\},\:1\leq r\leq s_{n}$.

Since the sequence $\left\{Y_{n}\right\}_{n\in\mathbb{T}_{+}}$ is convergent,
it is Cauchy, and there exists $n_{0}\in\mathbb{T}_{+}$ such
that
\begin{equation}
\label{prop_scal5} \left|\right. Y_{n}-Y_{n_{0}}\left|\right.< 2\sqrt{2},
\:\forall\,\,\,\, n\geq n_{0}
\end{equation}
Without loss of generality, assume $n_{0}=0$ (otherwise,
work with the sequence starting at $n_{0}$.) For every $n$, define
\begin{eqnarray}
\label{prop_scal6} \widetilde{s}_{n}&=&\min\left\{r\left|\right.i_{r}=1,\:1\leq r\leq s_{n}\right\}\\
\label{prop_scal7}
\widetilde{Y}_{n}&=&f_{i_{\widetilde{s}_{n}}}\circ\cdots\circ
f_{i_{s_{n}}}\left(P^{\ast}\right)
\end{eqnarray}
Clearly, $1+\sqrt{2}\leq\widetilde{Y}_{n}\leq 3$ and, by definition,
$\widetilde{Y}_{n}\in\mathcal{S}_{0}$, for all $n$. By basic
manipulations and using eqn.~(\ref{prop_scal5}), it can be shown
\begin{equation}
\label{prop_scal8} \widetilde{s}_{n}=\widetilde{s}_{0},\:\forall n
\end{equation}
We can then represent the sequence
$\left\{Y_{n}\right\}_{n\in\mathbb{T}_{+}}$ as
\begin{equation}
\label{prop_scal9}
Y_{n}=f_{0}^{\widetilde{s}_{0}-1}\left(\widetilde{Y}_{n}\right)
\end{equation}
Since $Y_{n}\rightarrow Y$, and the function
$f_{0}^{\widetilde{s}_{0}-1}(\cdot)$ is one-to-one and continuous, the
sequence $\left\{\widetilde{Y}_{n}\right\}_{n\in\mathbb{T}_{+}}$ must
converge to some $\widetilde{Y}$, i.e.,
\begin{equation}
\label{prop_scal10}
\lim_{n\rightarrow\infty}\widetilde{Y}_{n}=\widetilde{Y}
\end{equation}
It also follows that $1+\sqrt{2}\leq\widetilde{Y}\leq 3$ and
$\widetilde{Y}\in\mathcal{S}_{0}$, the set $\mathcal{S}_{0}$ being
closed. We then have
\begin{equation}
\label{prop_scal11}
Y=f_{0}^{\widetilde{s}_{0}-1}\left(\widetilde{Y}\right)
\end{equation}
which implies $Y\in\mathcal{S}_{\widetilde{s}_{0}-1}$. We thus showed
the inclusion
\begin{equation}
\label{prop_scal12}
\mbox{supp}\left(\mu^{\overline{\gamma}}\right)\subset\cup_{n\geq 0} \mathcal{S}_{n}
\end{equation}
The reverse inclusion is obvious, and we have the claim.
\end{proof}
Proposition~\ref{prop_scal} shows the self-similarity of
$\mbox{supp}\left(\mu^{\overline{\gamma}}\right)$ at scales of
$2^{n}$, where $n\in\mathbb{N}$. A rigorous definition of fractal
(see, for example, \cite{Mandelbrot}) requires self-similarity at
every scale, which we do not pursue here. This explains why we use
`fractal like'  to describe the structure of
$\mbox{supp}\left(\mu^{\overline{\gamma}}\right)$. The fractal
nature of $\mbox{supp}\left(\mu^{\overline{\gamma}}\right)$,
though not obvious, is not very surprising. In fact, it is known
(see, for example,~\cite{LasotaMackey,DiaconisFreedman}) that a
large class of iterated function systems (systems, which generate
a Markov process by random switching between a set of at most
countable functions)\footnote{The RRE can be viewed as an
iterated function system, where the iterations are randomly
switched between the Lyapunov $\left(f_{0}\right)$ and Riccati
$\left(f_{1}\right)$ functions.} leads to fractal invariant
distributions.

Apart from the holes (fractures) in
$\mbox{supp}\left(\mu^{\overline{\gamma}}\right)$ explained by
Proposition~\ref{prop_scal}, the set
$\mbox{supp}\left(\mu^{\overline{\gamma}}\right)$ contains much more
fractures, as observed in the numerical plots of
$\mbox{supp}\left(\mu^{\overline{\gamma}}\right)$
(see Fig.~\ref{plot_support1}.) 
  It follows from
Proposition~\ref{prop_scal} that a thorough study of
$\mbox{supp}\left(\mu^{\overline{\gamma}}\right)$ requires studying only one
of the sets, $\left\{\mathcal{S}_{n}\right\}_{n\in\mathbb{T}_{+}}$, as the
pattern is repeated over the real line.

In Fig.~\ref{plot_support1} on the top left, the blue region
corresponds to the set $\mathcal{S}_{0}$. The
figure shows that the set contains many fractures (these fractures
are internal to $\mathcal{S}_{0}$ and different from the holes
between consecutive elements of
$\left\{\mathcal{S}_{n}\right\}_{n\in\mathbb{T}_{+}}$ as explained by
Proposition~\ref{prop_scal}.) The blue blobs appearing in the
figure are fractured more finely, but the visualization software
limits the resolution by coalescing disconnected components
separated by small distances into one large blob. A better
visualization is obtained by looking at the set $\mathcal{S}_{1}$, see Fig.~\ref{plot_support1} on the top right,
which is a stretched version (by a factor of 2) of
$\mathcal{S}_{0}$, and more fractures are resolved.
\begin{figure}[ptb]
\begin{center}
\includegraphics[height=2.5in, width=2.5in]{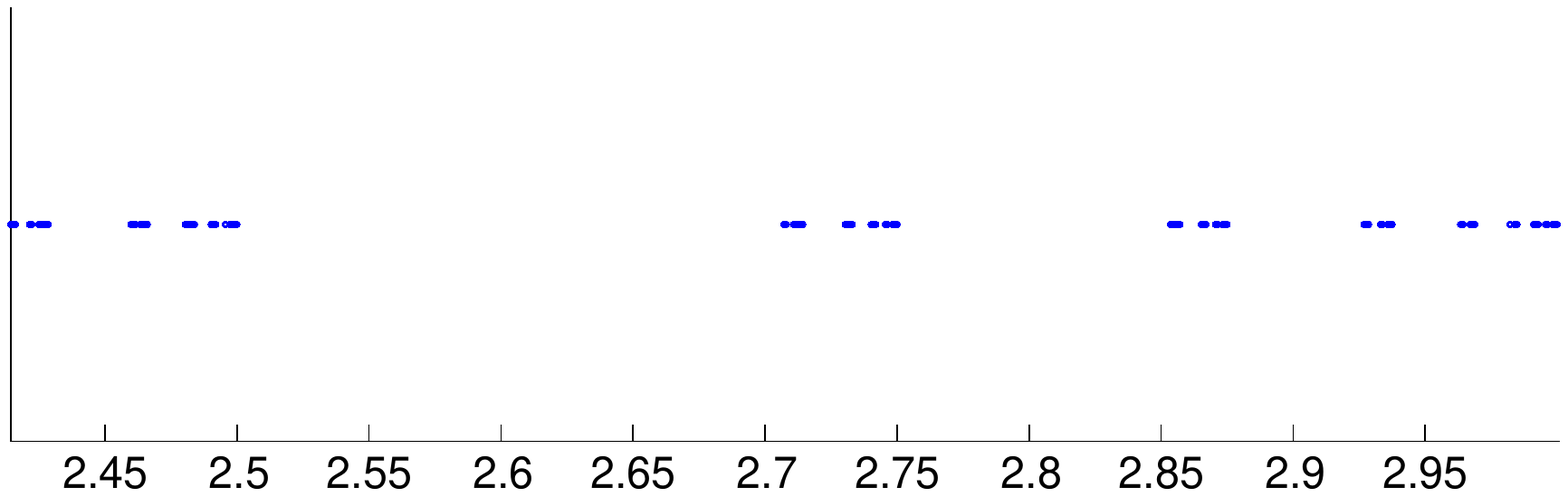}
\includegraphics[height=2.5in, width=3.5in]{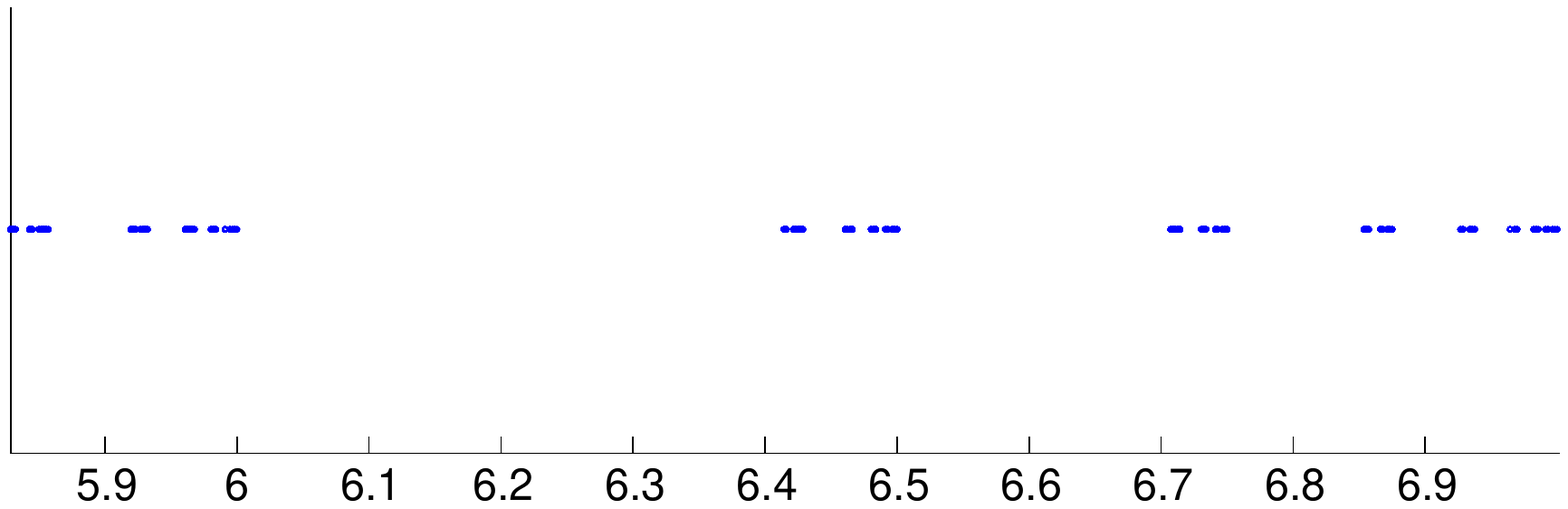}
\vspace{-1.5cm}
\includegraphics[height=3.5in, width=6.5in]{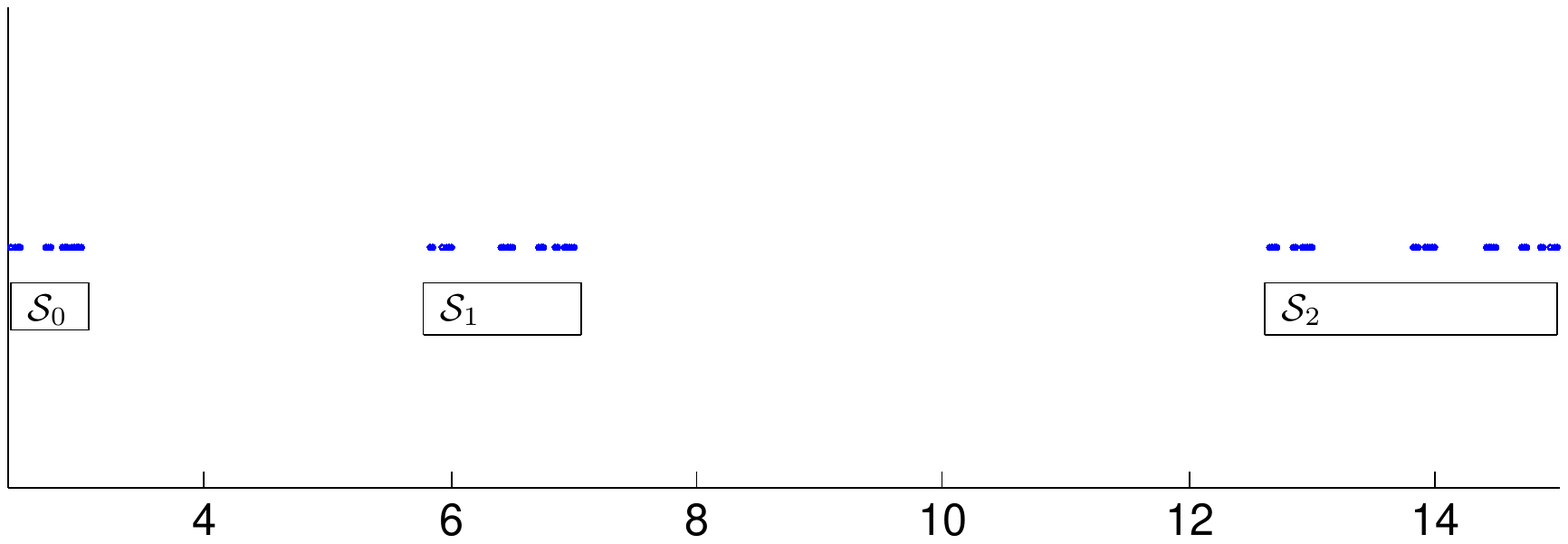}
\vspace{-1.5cm}\caption{Top left: Restriction of $\mu^{\overline{\gamma}}$ to
$[1+\sqrt{2},3]$, i.e., the set $\mathcal{S}_{0}$. Top right:
Restriction of $\mu^{\overline{\gamma}}$ to $[3+2\sqrt{2},7]$,
i.e., the set $\mathcal{S}_{1}$.
Bottom: Restriction of $\mu^{\overline{\gamma}}$ to
$[1+\sqrt{2},15]$, i.e., the set
$\mathcal{S}_{0}\cup\mathcal{S}_{1}\cup\mathcal{S}_{2}$.} \label{plot_support1}
\label{plot_support2}
\end{center}
\end{figure}
%
%
Finally, in Fig.~\ref{plot_support1} bottom, we plot
the set $\mu^{\overline{\gamma}}$ restricted to $[1+\sqrt{2},15]$,
i.e., the set
$\mathcal{S}_{0}\cup\mathcal{S}_{1}\cup\mathcal{S}_{2}$. The
figure demonstrates the self-similarity and the inter-set holes
(holes between consecutive elements of
$\{\mathcal{S}_{n}\}_{n\in\mathbb{T}_{+}}$) as explained by
Proposition~\ref{prop_scal}.

\subsection{Numerical Studies}
\label{num_stud} We study numerically the
eigenvalue distribution from $\mu^{\overline{\gamma}}$
($\overline{\gamma}^{\mbox{\scriptsize sb}}<\overline{\gamma}<1$) for a
10-dimensional system. The matrices $A$ and $C$ (of dimensions
$10\times 10$ and $5\times 10$, respectively) are generated
randomly, so that the assumptions of Theorem~\ref{main_th} are
satisfied.

In Fig.~\ref{eig_l} on the left, we plot the cumulative distribution
function (c.d.f.) of the largest eigenvalue $\lambda_{10}(\cdot)$
from $\mu^{\overline{\gamma}}$ for different values of
$\overline{\gamma}^{\mbox{\scriptsize sb}}<\overline{\gamma}<1$. From the figure, we see (as expected) that, as $\overline{\gamma}$ increases to 1, the
distributions approach the Dirac distribution
$\delta_{\lambda_{10}\left(P^{\ast}\right)}$ (the distribution with entire
mass concentrated at $\lambda_{10}\left(P^{\ast}\right)$, the largest
eigenvalue of the deterministic fixed point $P^{\ast}$ of the
algebraic Riccati equation.) The fact that, for
$\overline{\gamma}^{\mbox{\scriptsize sb}}<\overline{\gamma}<1$, the support of the eigenvalue distribution is a subset of
$\left[\lambda_{10}\left(P^{\ast}\right),\infty\right)$ follows from
Theorem~\ref{supp_inv} (eqn.~(\ref{main_th1}).

Similarly, in Fig.~(\ref{eig_l}) on the right, we plot the
c.d.f.~of the trace (which is
the conditional mean-squared error) from $\mu^{\overline{\gamma}}$
for different values of
$\overline{\gamma}^{\mbox{\scriptsize sb}}<\overline{\gamma}<1$. From the figure, we
see (as expected) that, as $\overline{\gamma}$ increases to 1, the
distributions approach the Dirac distribution
$\delta_{\mbox{Tr}\left(P^{\ast}\right)}$.

The eigenvalue distributions can be used in system design for control and estimation problems. Since the RRE converges in distribution to
$\mu^{\overline{\gamma}}$, the system designer can tune the
operating $\overline{\gamma}$ to ensure satisfactory system
performance.
\begin{figure}[ptb]
\begin{center}
\includegraphics[height=2.5in, width=3in]{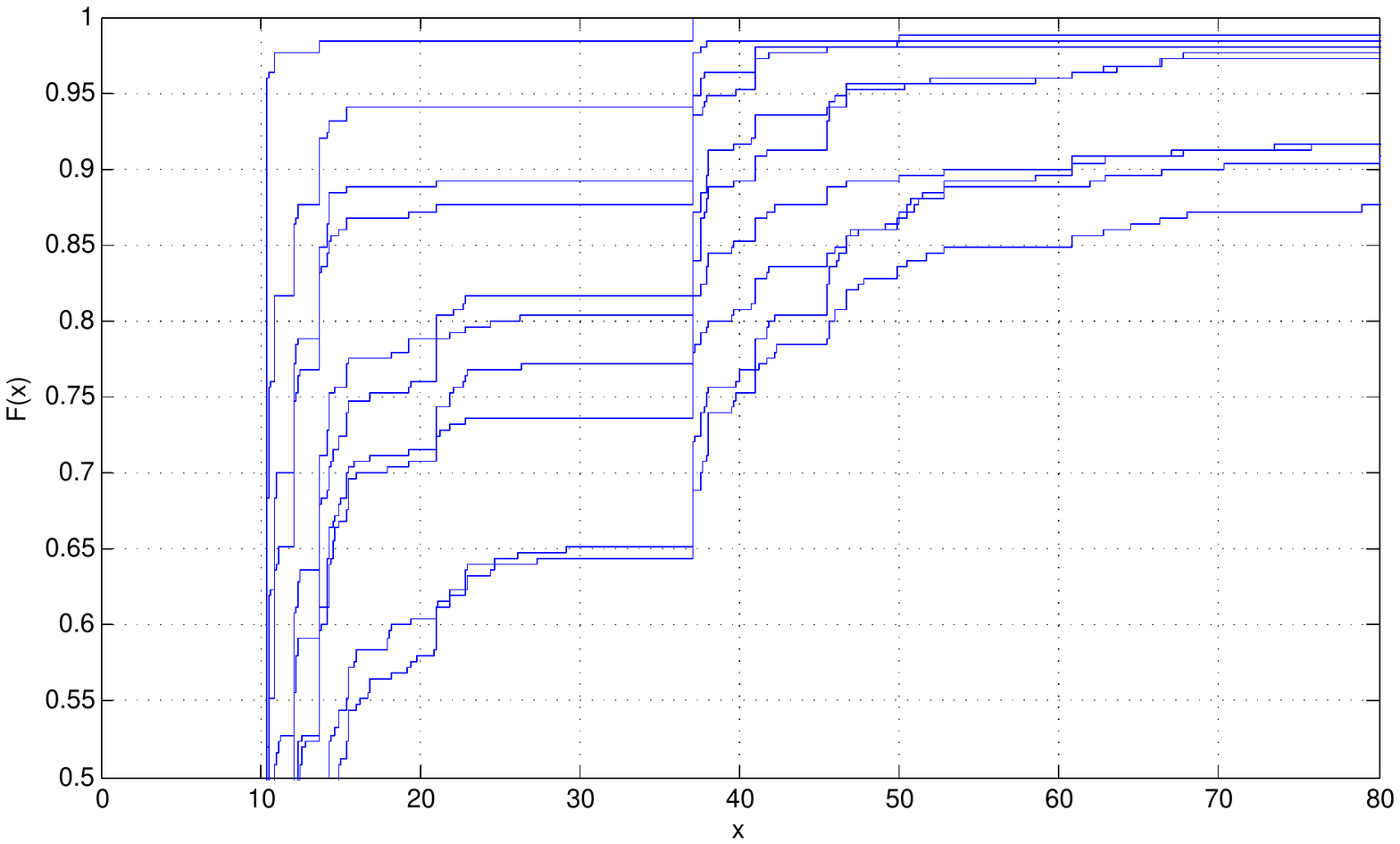}
\includegraphics[height=2.5in, width=3in]{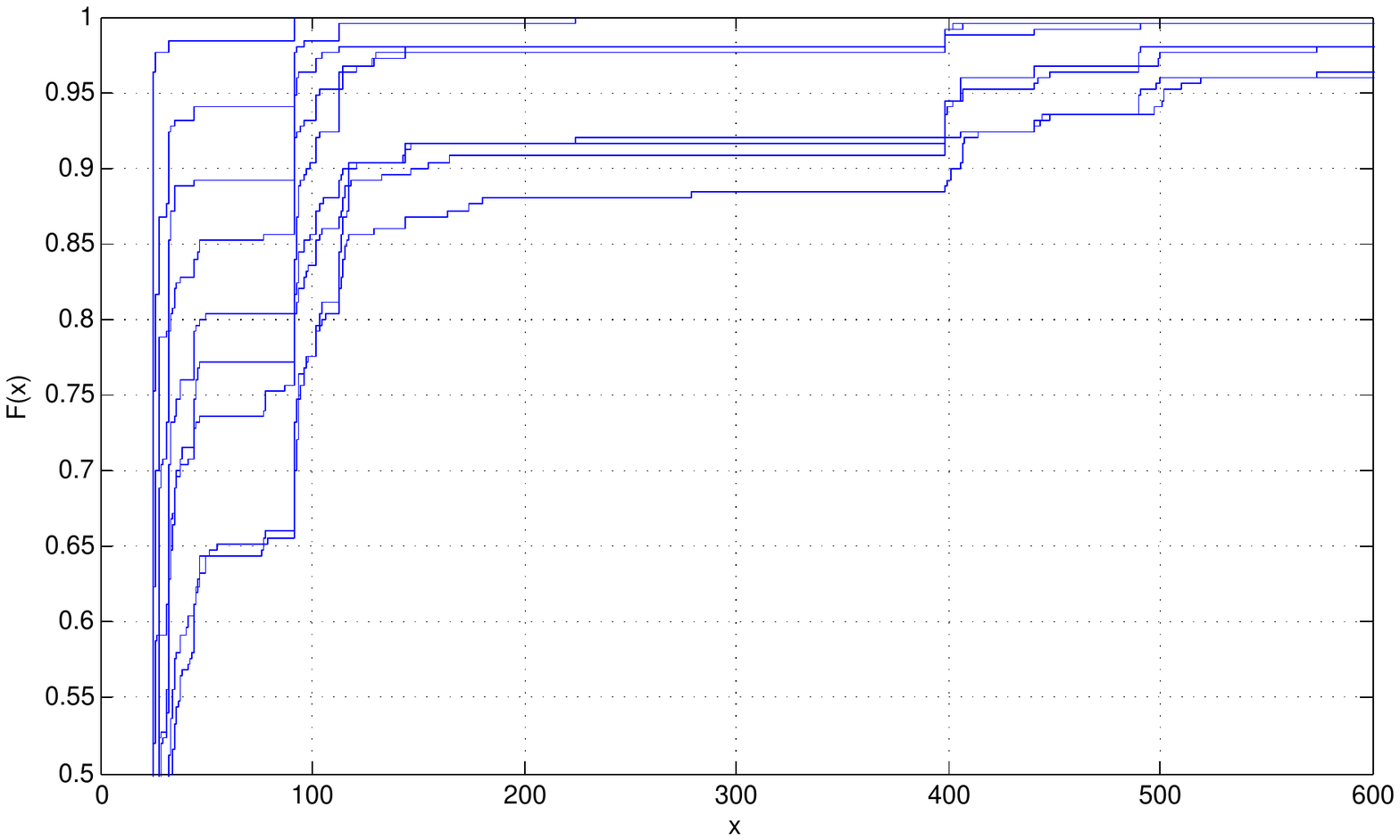}
\vspace{-2cm}\caption{Left: CDF of largest eigenvalue from
$\mu^{\overline{\gamma}}$ for varying $\overline{\gamma}$, as
$\overline{\gamma}$ approaches 1 ($\overline{\gamma}$ increases from right to left.) Right: CDF of trace from
$\mu^{\overline{\gamma}}$ for varying $\overline{\gamma}$, as
$\overline{\gamma}$ approaches 1 ($\overline{\gamma}$ increases from right to left.)} \label{eig_l}
\end{center}
\end{figure}

We end this subsection by remarking on the numerical computation
of moments from the invariant distribution
$\mu^{\overline{\gamma}}$
$\left(\overline{\gamma}^{\mbox{\scriptsize
sb}}<\overline{\gamma}<1.\right)$ We set notation first.

Define $L^{1}\left(\mu^{\overline{\gamma}}\right)$ to be the set of
integrable functions on $\mathbb{S}^{N}_{+}$ w.r.t.~the measure
$\mu^{\overline{\gamma}}$. We then have the following assertion:
\begin{proposition}
\label{prop_sim} Fix $\overline{\gamma}^{\mbox{\scriptsize sb}}<\overline{\gamma}<1$
and let $h\in L^{1}\left(\mu^{\overline{\gamma}}\right)$. Then there exists a
set $\mathcal{S}_{h}^{\overline{\gamma}}\subset\mathbb{S}^{N}_{+}$
with
$\mu^{\overline{\gamma}}\left(\mathcal{S}_{h}^{\overline{\gamma}}\right)=1$,
such that, for every
$P_{0}\in\mathcal{S}_{h}^{\overline{\gamma}}$, the following
holds:
\begin{equation}
\label{prop_sim1}
\lim_{t\rightarrow\infty}\frac{1}{t}\sum_{s=0}^{t-1}h(P_{t})=
\int_{\mathbb{S}^{N}_{+}}h(Y)\mu^{\overline{\gamma}}d(Y),
\:\mathbb{P}^{\overline{\gamma},P_{0}}\:\mbox{a.s.}
\end{equation}
where $\left\{P_{t}\right\}_{t\in\mathbb{T}_{+}}$ is the sequence generated
by the RRE with initial condition $P_{0}$.
\end{proposition}
\begin{proof}
It follows from the fact that $\mu^{\overline{\gamma}}$ is an
ergodic measure (being attractive) for the transition probability
operator $\mathbb{Q}^{\overline{\gamma}}$ of the Markov process
$\left\{P_{t}\right\}_{t\in\mathbb{T}_{+}}$ (see, for
example, \cite{LermaLasserre}.)
\end{proof}
Proposition~\ref{prop_sim} has important consequences in computing
moments and probabilities (for example, the probability of escape
from a set can be obtained by using $h$ to be the complementary
indicator function) from the invariant distribution
$\mu^{\overline{\gamma}}$. It says that, for a function
$h:\mathbb{S}^{N}_{+}\longmapsto\mathbb{R}$ with finite
$\mu^{\overline{\gamma}}$-moment, there exists a set
$\mathcal{S}_{h}^{\overline{\gamma}}$ with
$\mu^{\overline{\gamma}}$-probability one, such that, if the
initial condition belongs to
$\mathcal{S}_{h}^{\overline{\gamma}}$, the empirical moment
converges to the $\mu^{\overline{\gamma}}$-moment for every sample
path a.s. Thus, in order to compute a
$\mu^{\overline{\gamma}}$-moment, generating a single instance of
the Markov process suffices, as long as the initial condition
belongs to the set $\mathcal{S}_{h}^{\overline{\gamma}}$. This has
important consequences in moment computation from the invariant
distribution as, under the assumptions of
Proposition~\ref{prop_sim}, one does not need to run costly
simulations to generate the distribution $\mu^{\overline{\gamma}}$
empirically; rather, the generation of a single path would
suffice.

The difficulty in using Proposition~\ref{prop_sim} is that the set
$\mathcal{S}_{h}^{\overline{\gamma}}$ is known only up to
a.s.~equivalence and further depends on $h$. In general,
$\mathcal{S}_{h}^{\overline{\gamma}}$ is not the entire
$\mathbb{S}^{N}_{+}$.\footnote{In fact,
$\mathcal{S}_{h}^{\overline{\gamma}}=\mathbb{S}^{N}_{+}$
\emph{iff} the Markov process is positive Harris recurrent, a
property that iterated function systems do not possess generally
(see~\cite{LermaLasserre}.)} However, Theorem~\ref{supp_inv}
provides us with the support of $\mu^{\overline{\gamma}}$ and
implies, in particular, that for every $\varepsilon>0$ and
$Y\in\mathcal{S}$ the open ball $B_{\varepsilon}(Y)$ has positive
measure. Thus, by choosing initial conditions $P_{0}$ arbitrarily
close to (and including) a $Y\in\mathcal{S}$, one is likely to get
the convergence in Proposition~\ref{prop_sim}.

Also, we note that the set of functions $h\in L^{1}(\mathbb{\mu}^{\overline{\gamma}})$ is not empty. As a matter of fact, all bounded measurable functions $h:\mathbb{S}_{+}^{N}\longmapsto\mathbb{R}$ are contained in $L^{1}(\mathbb{\mu}^{\overline{\gamma}})$, for every $\overline{\gamma}$ which guarantees the existence and uniqueness of $\mathbb{\mu}^{\overline{\gamma}}$ (for example, $\overline{\gamma}>0$ for stabilizable and detectable systems.) In some situations, it may be possible to determine the moments of unbounded functionals by approximating them by a sequence of suitable truncations and then invoking some form of dominated convergence. An important case is the mean evaluation corresponding to $h(Y)=Y$. In that case, by Theorem~\ref{main_res} we note that if we operate at $\overline{\gamma}>\overline{\gamma}^{\mbox{\scriptsize{bim}}}$, the integral w.r.t. the corresponding $\mathbb{\mu}^{\overline{\gamma}}$ exists and hence, one may invoke Proposition~\ref{prop_sim} to compute the mean under the invariant measure.

\section{Conclusions and future work}
\label{conclusion} The paper presents a new
analysis of the Random Riccati Equation. It studies the evolution
of the state error covariance arising from a Kalman filter where
observations can be lost according to an i.i.d.~Bernoulli process. This process
can model an analog erasure channel between the sensor and the
estimator. Following a novel approach based on random dynamical systems, we
provide an exhaustive analysis of the steady state
behavior of the filter.

We showed the existence of critical arrival probabilities
$\overline{\gamma}^{\mbox{\scriptsize sb}}, \overline{\gamma}^{\mbox{\scriptsize bim}}$, such that the
error covariance converges in distribution to a unique steady
state distribution if the arrival probability
$\overline{\gamma}>\overline{\gamma}^{\mbox{\scriptsize sb}}$; this distribution has finite mean
for $\overline{\gamma}>\overline{\gamma}^{\mbox{\scriptsize bim}}$. Additionally, we
provided a characterization of the support of the steady state
distribution, showing its fractal characteristics. The latter
result, combined with ergodicity arguments, provides a method to
numerically evaluate the error covariance steady state
distribution. We feel that our approach is particularly amenable
to addressing general problems of networked control problems, as
they provide a theoretical framework to combine stochastic
processes used in the modeling of communication networks with
differential equations describing the evolution of dynamical
systems.

The approach and results in the paper will easily extend to problems of control over erasure channels. Further, we plan to describe more complex
interactions and tradeoffs between communication and control via
random dynamical systems.

\appendices

\section{Proofs of
Propositions~\ref{prop_MarkFell},\ref{imp_sb},\ref{prop_crel}}
\label{proof_props}
{\small
\begin{proof}[Proposition~\ref{prop_MarkFell}]
The fact that $T^{\overline{\gamma}}$ is a Markov operator is a
standard consequence of $\mathcal{Q}^{\overline{\gamma}}$ being a
transition probability (see, for example, \cite{Zaharopol}.) Also, $L^{\overline{\gamma}}$ is linear.
Thus, we only need to verify that $L^{\overline{\gamma}}$ maps
$C_{b}\left(\mathbb{S}_{+}^{N}\right)$ to
$C_{b}\left(S_{n}^{+}\right)$. For linear operators
generated by eqn.~(\ref{def_Lgamma}), such a property is called
the weak-Feller property of the transition probability
$\mathbb{Q}^{\overline{\gamma}}$ (see~\cite{LermaLasserre}.) It
can be shown (see Proposition 7.2.1.~in~\cite{LermaLasserre}) that
$\mathbb{Q}^{\overline{\gamma}}$ is weak-Feller \emph{iff} for
every sequence $\left\{Y_{n}\right\}_{n\in\mathbb{T}_{+}}$ in
$\mathbb{S}^{N}_{+}$ that converges to some
$Y\in\mathbb{S}^{N}_{+}$, and every open set
$\mathcal{O}\in\mathcal{B}\left(\mathbb{S}^{N}_{+}\right)$
\begin{equation}
\label{prop_MarkFell4}
\liminf_{n\rightarrow\infty}\mathbb{Q}^{\overline{\gamma}}
\left(Y_{n},\mathcal{O}\right)\geq
\mathbb{Q}^{\overline{\gamma}}\left(Y,\mathcal{O}\right)
\end{equation}
To verify eqn.~(\ref{prop_MarkFell4}), we note that by
eqn.~(\ref{def_trQ})
\begin{equation}
\label{prop_MarkFell5}
\mathbb{Q}^{\overline{\gamma}}\left(Y_{n},\mathcal{O}\right)=
\left(1-\overline{\gamma}\right)\mathbb{I}_{\mathcal{O}}
\left(f_{0}\left(Y_{n}\right)\right)+\overline{\gamma}\mathbb{I}_{\mathcal{O}}
\left(f_{1}\left(Y_{n}\right)\right),\:\forall
n\in\mathbb{T}_{+}
\end{equation}
If $f_{0}(Y)\in\mathcal{O}$, from the continuity of
$f_{0}$ (Lemma~\ref{prop1}) and  $\mathcal{O}$
open, $\exists\, n_{0}\in\mathbb{T}_{+}$, s.t.
\begin{equation}
\label{prop_MarkFell6} f_{0}\left(Y_{n}\right)\in\mathcal{O},\:n\geq n_{0}
\end{equation}
which implies
\begin{eqnarray}
\label{prop_MarkFell7}
\lim_{n\rightarrow\infty}\mathbb{I}_{\mathcal{O}}
\left(f_{0}\left(Y_{n}\right)\right)
=1 
=\mathbb{I}_{\mathcal{O}}\left(f_{0}(Y)\right)
\end{eqnarray}
On the other hand, if $f_{0}(Y)\notin\mathcal{O}$, we have
\begin{eqnarray}
\label{prop_MarkFell8}
\liminf_{n\rightarrow\infty}\mathbb{I}_{\mathcal{O}}
\left(f_{0}\left(Y_{n}\right)\right)
\geq 0 
=\mathbb{I}_{\mathcal{O}}\left(f_{0}(Y)\right)
\end{eqnarray}
Similarly, we have
\begin{equation}
\label{prop_MarkFell9}
\liminf_{n\rightarrow\infty}\mathbb{I}_{\mathcal{O}}
\left(f_{1}\left(Y_{n}\right)\right)\geq
\mathbb{I}_{\mathcal{O}}\left(f_{1}(Y)\right)
\end{equation}
and eqn.~(\ref{prop_MarkFell4}) follows.
\end{proof}

\begin{proof}[Proposition~\ref{imp_sb}]
Assume $\left\{P_{t}\right\}_{t\in\mathbb{T}_{+}}$ is b.i.m.~for some
$\overline{\gamma}$ and $P_{0}$, i.e., $\exists\,
M^{\overline{\gamma},P_{0}}$ such that
\begin{equation}
\label{imp_sb3}
\sup_{t\in\mathbb{T}_{+}}\mathbb{E}^{\overline{\gamma},P_{0}}
\left[P_{t}\right]\preceq M^{\overline{\gamma},P_{0}}
\end{equation}
By the positive semidefinitess of the matrices
and  linearity of the trace,
\begin{eqnarray}
\label{imp_sb4}
\forall t\in\mathbb{T}_{+}:\: \mathbb{E}^{\overline{\gamma},P_{0}}\left[\left\|P_{t}\right\|\right]
\leq \mathbb{E}^{\overline{\gamma},P_{0}}\left[\mbox{Tr}\,P_{t}\right]
= \mbox{Tr}\left(\mathbb{E}^{\overline{\gamma},P_{0}}
\left[P_{t}\right]\right)
\leq \mbox{Tr}\,M^{\overline{\gamma},P_{0}}
\end{eqnarray}
Chebyshev's inequality then implies
\begin{eqnarray}
\label{imp_sb5}
\mathbb{P}^{\overline{\gamma},P_{0}}\left(\left\|P_{t}\right\|>N\right)&\leq&
\frac{\mbox{Tr}\,M^{\overline{\gamma},P_{0}}}{N}\\
\label{imp_sb6}
\limsup_{N\rightarrow\infty}\sup_{t\in\mathbb{T}_{+}}
\mathbb{P}^{\overline{\gamma},P_{0}}\left(\left\|P_{t}\right\|>N\right)
& \leq &
\lim_{N\rightarrow\infty}\frac{\mbox{Tr}\,M^{\overline{\gamma},P_{0}}}{N}\nonumber
=0
\end{eqnarray}
\end{proof}

\begin{proof}[Proposition~\ref{prop_crel}]
Part~i) is obvious from Proposition~\ref{imp_sb}.

For part~ii) we consider the case of unstable $A$. For stable $A$, the proposition is trivial and follows from the fact, that, the unconditional variance of the state sequence reaches a steady state (hence bounded), and a suboptimal estimate $\widehat{x}_{t}\equiv 0$ for all $t$ leads to pathwise boundedness of the corresponding error covariance. In fact, as pointed earlier, in this case, even $\overline{\gamma}=0$ leads to stochastic boundedness of the sequence $\{P_{t}\}$ from every initial condition.

We now prove part [ii] for unstable $A$, under the additional assumption of invertibility of $C$ (the general case being considered in Lemma 13 of~\cite{Riccati-moddev}.) Fix $\overline{\gamma}>0$,
$P_{0}\in\mathbb{S}^{N}_{+}$ and recall the function
$f_{1}(\cdot)$ (eqn.~(\ref{def_f1}))
\begin{equation}
\label{prop_crel4}
f_{1}(X)=AXA^{T}+Q-AXC^{T}\left(CXC^{T}+R\right)^{-1}CXA^{T}
\end{equation}
Since $C$ is invertible, we use the matrix inversion lemma to
obtain
\begin{equation}
\label{prop_crel5} \left(CXC^{T}+R\right)^{-1} =
C^{-T}X^{-1}C^{-1}-C^{-T}X^{-1}C^{-1}
\left(R^{-1}+C^{-T}X^{-1}C^{-1}\right)^{-1}C^{-T}X^{-1}C^{-1}
\end{equation}
By substituting into eqn.~(\ref{prop_crel4}), we have
\begin{equation}
\label{prop_crel6}
f_{1}(X)=AC^{-1}\left(R^{-1}+C^{-T}X^{-1}C^{-1}\right)^{-1}C^{-T}A^{T}+Q
\end{equation}
Using
$
\left\|\left(R^{-1}+C^{-T}X^{-1}C^{-1}\right)^{-1}
\right\|\leq\left\|R\right\|
$,
we have
\begin{equation}
\label{prop_crel8}
\left\|f_{1}(X)\right\|\leq\|A\|\|C\|\|C^{-T}\|\|A^{T}\|\|R\|+\left\|Q\right\|
\end{equation}
For $N\in\mathbb{T}_{+}$ and sufficiently large, define
\begin{equation}
\label{prop_crel9}
k(N)=\min\left\{k\in\mathbb{T}_{+}\left|\right.M
\left(\alpha^{2(k-1)}-1\right)+\frac{\left\|Q\right\|\alpha^{2(k-1)}}{\alpha^{2}-1}\geq
N\right\}
\end{equation}
where $\alpha$ is the absolute value of the largest eigenvalue of
$A$  and  $M$ is
\begin{equation}
\label{prop_crel10}
M=\max\left\{\|A\|\|C\|\|C^{-T}\|\|A^{T}\|\|R\|
+\left\|Q\right\|,\left\|P_{0}\right\|\right\}
\end{equation}
As $N\uparrow\infty$, $k(N)\uparrow\infty$. To estimate
$\mathbb{P}^{\overline{\gamma},P_{0}}\left(\left\|P_{t}\right\|>N\right)$,
 define the random time $\widetilde{t}$ by
\begin{equation}
\label{prop_crel11} \widetilde{t}=\max\left\{0< s\leq
t\left|\right.\gamma_{s-1}=1\right\}
\end{equation}
where the maximum of an empty set is taken to be zero. Using the
above arguments, we have
\begin{equation}
\label{prop_crel12} \left\|P_{\widetilde{t}}\right\|\leq M
\end{equation}
Indeed, if $\widetilde{t}=0$, eqn.~(\ref{prop_crel12}) clearly holds by
definition of $M$. On the contrary, if $\widetilde{t}>0$, we have
\begin{eqnarray}
\label{prop_crel13} \left\|P_{\widetilde{t}}\right\| 
=\left\|f_{1}\left(P_{\widetilde{t}-1}\right)\right\|
\leq \|A\|\|C\|\|C^{-T}\|\|A^{T}\|\|R\|
+\left\|Q\right\|
\leq M
\end{eqnarray}
We also have
$
P_{t} =
f_{0}^{t-\widetilde{t}-1}\left(P_{\widetilde{t}}\right)
$,
which implies
\begin{eqnarray}
\nonumber
\left\|P_{t}\right\|
=
\left\|f_{0}^{t-\widetilde{t}-1}\left(P_{\widetilde{t}}\right)\right\|
\leq
M\alpha^{2(t-\widetilde{t}-1)}+\left\|Q\right\|\sum_{k=1}^{t-\widetilde{t}-1}
\alpha^{2(k-1)}
=
M\alpha^{2(t-\widetilde{t}-1)}+
\left\|Q\right\|\frac{\alpha^{2(t-\widetilde{t}-1)}-1}{\alpha^{2}-1}
\end{eqnarray}
From this inequality and eqn.~(\ref{prop_crel9}), 
  it follows
\begin{equation}
\label{prop_crel16}
\mathbb{P}^{\overline{\gamma},P_{0}}\left(\left\|P_{t}\right\|>N\right)\leq
\mathbb{P}^{\overline{\gamma},P_{0}}\left(t-\widetilde{t}\geq
k(N)\right)
\end{equation}
We observe that the random time 
$
0\leq\widetilde{t}\leq t
$.
In case, $0\leq k\leq t-1$,
\begin{eqnarray}
\label{prop_crel18}
\mathbb{P}^{\overline{\gamma},P_{0}}\left(t-\widetilde{t}=k\right)
= \mathbb{P}^{\overline{\gamma},P_{0}}\left(\gamma_{t-k-1}=1,
\{\gamma_{s}=0\}_{t-k\leq s<t}\right)
=
\overline{\gamma}\left(1-\overline{\gamma}\right)^{k}
\leq \left(1-\overline{\gamma}\right)^{k}
\end{eqnarray}
If $k=t$, we have
\begin{eqnarray}
\label{prop_crel19}
\mathbb{P}^{\overline{\gamma},P_{0}}\left(t-\widetilde{t}=k\right)
=
\mathbb{P}^{\overline{\gamma},P_{0}}\left(\{\gamma_{s}=0\}_{0\leq
s<t}\right)
= \left(1-\overline{\gamma}\right)^{k}
\end{eqnarray}
We thus have the upper bound (possibly loose, but sufficient for our purposes)
\begin{eqnarray}
\label{prop_crel20}
\mathbb{P}^{\overline{\gamma},P_{0}}\left(t-\widetilde{t}\geq
k(N)\right)
\leq
\sum_{k=k(N)}^{\infty}\left(1-\overline{\gamma}\right)^{k}
=
\frac{\left(1-\overline{\gamma}\right)^{k(N)}}{\overline{\gamma}}
\end{eqnarray}
%
From eqns.~(\ref{prop_crel16},\ref{prop_crel20}), we have for all
$t$ and sufficiently large~$N$
\begin{equation}
\label{prop_crel21}
\mathbb{P}^{\overline{\gamma},P_{0}}\left(\left\|P_{t}\right\|>N\right)
\leq\frac{\left(1-\overline{\gamma}\right)^{k(N)}}{\overline{\gamma}}
\end{equation}
Since $\overline{\gamma}>0$ and $k(N)\uparrow\infty$ as
$N\uparrow\infty$, it then follows
\begin{equation}
\label{prop_crel22}
\lim_{N\rightarrow\infty}\sup_{t\in\mathbb{T}_{+}}
\mathbb{P}^{\overline{\gamma},P_{0}}\left(\left\|P_{t}\right\|>N\right)=0
\end{equation}
Thus, $\left\{P_{t}\right\}_{t\in\mathbb{T}_{+}}$ is s.b. (for all initial
conditions $P_{0}$) for every $\overline{\gamma}>0$. The
Proposition follows.
\end{proof}
}

\bibliographystyle{IEEEtran}
\bibliography{IEEEabrv,CentralBib}

\end{document}